\pdfoutput=1
\documentclass[11pt]{article}
\usepackage[protrusion=true,expansion=true]{microtype}
\usepackage{amsmath,amssymb,amsfonts,amsthm}
\usepackage{subcaption}
\usepackage{graphicx}
\usepackage{fullpage}
\usepackage[backref=page]{hyperref}
\usepackage{color}
\usepackage{wrapfig}
\usepackage{tikz}
\usetikzlibrary{decorations.pathreplacing}
\usepackage{setspace}
\usepackage{algorithm}
\usepackage[noend]{algpseudocode}
\usepackage[framemethod=tikz]{mdframed}
\usepackage{xspace}
\usepackage{pgfplots}
\usepackage{framed}
\usepackage{subcaption}
\usepackage{thmtools}
\usepackage{thm-restate}
\pgfplotsset{compat=1.5}

\newtheorem{theorem}{Theorem}[section]
\newtheorem{corollary}[theorem]{Corollary}
\newtheorem{lemma}[theorem]{Lemma}

\newtheorem{definition}[theorem]{Definition}
\newtheorem{remark}[theorem]{Remark}

\newtheorem{problem}[theorem]{Problem}
\newtheorem{assumption}[theorem]{Assumption}

\newenvironment{proofof}[1]{\begin{trivlist} \item {\bf Proof
#1:~~}}
  {\qed\end{trivlist}}

\newcommand{\namedref}[2]{\hyperref[#2]{#1~\ref*{#2}}}
\newcommand{\thmlab}[1]{\label{thm:#1}}
\newcommand{\thmref}[1]{\namedref{Theorem}{thm:#1}}
\newcommand{\lemlab}[1]{\label{lem:#1}}
\newcommand{\lemref}[1]{\namedref{Lemma}{lem:#1}}

\newcommand{\corlab}[1]{\label{cor:#1}}

\newcommand{\seclab}[1]{\label{sec:#1}}
\newcommand{\secref}[1]{\namedref{Section}{sec:#1}}

\newcommand{\alglab}[1]{\label{alg:#1}}
\renewcommand{\algref}[1]{\namedref{Algorithm}{alg:#1}}

\newcommand{\deflab}[1]{\label{def:#1}}
\newcommand{\defref}[1]{\namedref{Definition}{def:#1}}

\newcommand{\assumref}[1]{\namedref{Assumption}{assum:#1}}
\newcommand{\assumlab}[1]{\label{assum:#1}}

\def \HAM    {\mdef{\mathsf{HAM}}}
\def \detgapeq    {\mdef{\textsc{DetGapEQ}}}
\def \bernmg    {\mdef{\textsc{BernMG}}}
\def \bernhhh    {\mdef{\textsc{BernHHH}}}
\def \lzero    {\mdef{\textsc{Estimate-L0}}}
\def \misragries    {\mdef{\textsc{MisraGries}}}
\def \Gen    {\mdef{\mathsf{Gen}}}
\def \Eval    {\mdef{\mathsf{Eval}}}

\def \oreq    {\mdef{\textsc{OrEq}}}
\def \streamalg    {\mdef{\mathsf{StreamAlg}}}
\def \adversary    {\mdef{\mathsf{Adversary}}}

\newcommand{\PPPr}[2]{\ensuremath{\underset{#1}{\mathbf{Pr}}\left[#2\right]}}

\renewcommand{\O}[1]{\ensuremath{\mathcal{O}\left(#1\right)}}
\newcommand{\tO}[1]{\ensuremath{\tilde{\mathcal{O}}\left(#1\right)}}
\newcommand{\eps}{\varepsilon}

\def \calA    {\mdef{\mathcal{A}}}

\def \calQ    {\mdef{\mathcal{Q}}}

\def \bf    {\mdef{\mathbf{f}}}

\newcommand{\mdef}[1]{{\ensuremath{#1}}\xspace}  
\newcommand{\myfunc}[1]{\mdef{\mathsf{#1}}}      

\DeclareMathOperator*{\negl}{negl}

\DeclareMathOperator*{\poly}{poly}

\def \negl     {\mdef{\myfunc{negl}}}                

\newcommand{\ignore}[1]{}

\newif\ifnotes\notestrue 
\ifnotes
\newcommand{\samson}[1]{\textcolor{purple}{{\bf (Samson:} {#1}{\bf ) }} \marginpar{\tiny\bf
             \begin{minipage}[t]{0.5in}
               \raggedright S:
            \end{minipage}}}            							
\else
\newcommand{\samson}[1]{}
\fi

\makeatletter
\renewcommand*{\@fnsymbol}[1]{\textcolor{mahogany}{\ensuremath{\ifcase#1\or *\or \dagger\or \ddagger\or
 \mathsection\or \triangledown\or \mathparagraph\or \|\or **\or \dagger\dagger
   \or \ddagger\ddagger \else\@ctrerr\fi}}}
\makeatother

\providecommand{\email}[1]{\href{mailto:#1}{\nolinkurl{#1}\xspace}}

\definecolor{mahogany}{rgb}{0.75, 0.25, 0.0}
\definecolor{darkblue}{rgb}{0.0, 0.0, 0.55}
\definecolor{darkpastelgreen}{rgb}{0.01, 0.75, 0.24}
\definecolor{darkgreen}{rgb}{0.0, 0.2, 0.13}
\definecolor{bleudefrance}{rgb}{0.19, 0.55, 0.91}
\definecolor{forestgreen(web)}{rgb}{0.13, 0.55, 0.13}
\hypersetup{
     colorlinks   = true,
     citecolor    = mahogany,
	 linkcolor	  = forestgreen(web),
	 urlcolor     = mahogany
}

\usepackage[capitalise]{cleveref}

\newcommand{\te}{\text}

\newcommand{\cA}{\mathcal A}

\begin{document}

\title{The White-Box Adversarial Data Stream Model
}

\author{
Mikl\'{o}s Ajtai\thanks{Hungarian Academy of Sciences. E-mail: \email{miklos.ajtai@gmail.com}} 
\and
Vladimir Braverman\thanks{Google Research. E-mail: \email{vbraverman@google.com}}
\and
T.S. Jayram\thanks{Lawrence Livermore National Laboratories. E-mail: \email{t.s.jayram@gmail.com}}
\and
Sandeep Silwal\thanks{MIT. E-mail: \email{silwal@mit.edu}}
\and
Alec Sun\thanks{Carnegie Mellon University. E-mail: \email{alecsun@andrew.cmu.edu}}
\and
David P. Woodruff\thanks{Carnegie Mellon University. E-mail: \email{dwoodruf@andrew.cmu.edu}}
\and
Samson Zhou\thanks{Carnegie Mellon University. E-mail: \email{samsonzhou@gmail.com}}
}

\maketitle

\begin{abstract}
There has been a flurry of recent literature studying streaming algorithms for which the input stream is chosen adaptively by a black-box adversary who observes the output of the streaming algorithm at each time step. However, these algorithms fail when the adversary has access to the internal state of the algorithm, rather than just the output of the algorithm. 

We study streaming algorithms in the \emph{white-box adversarial model}, where the stream is chosen adaptively by an adversary who observes the entire internal state of the algorithm at each time step. We show that nontrivial algorithms are still possible. We first give a randomized algorithm for the $L_1$-heavy hitters problem that outperforms the optimal deterministic Misra-Gries algorithm on long streams. If the white-box adversary is computationally bounded, we use cryptographic techniques to reduce the memory of our $L_1$-heavy hitters algorithm even further and to design a number of additional algorithms for graph, string, and linear algebra problems. The existence of such algorithms is surprising, as the streaming algorithm does not even have a secret key in this model, i.e., its state is entirely known to the adversary. One algorithm we design is for estimating the number of distinct elements in a stream with insertions and deletions achieving a multiplicative approximation and sublinear space; such an algorithm is impossible for deterministic algorithms. 

We also give a general technique that translates any two-player {\it deterministic} communication lower bound to a lower bound for {\it randomized} algorithms robust to a white-box adversary. In particular, our results show that for all $p\ge 0$, there exists a constant $C_p>1$ such that any $C_p$-approximation algorithm for $F_p$ moment estimation in insertion-only streams with a white-box adversary requires $\Omega(n)$ space for a universe of size $n$. Similarly, there is a constant $C>1$ such that any $C$-approximation algorithm in an insertion-only stream for matrix rank requires $\Omega(n)$ space with a white-box adversary. These results do not contradict our upper bounds since {\it they assume the adversary has unbounded computational power}. Our algorithmic results based on cryptography thus show a separation between computationally bounded and unbounded adversaries. 

Finally, we prove a lower bound of $\Omega(\log n)$ bits for the fundamental problem of deterministic approximate counting in a stream of 0's and 1's, which holds even if we know how many total stream updates we have seen so far at each point in the stream. Such a lower bound for approximate counting with additional information was previously unknown, and in our context, it shows a separation between multiplayer deterministic maximum communication and the white-box space complexity of a streaming algorithm. 
\end{abstract}

\section{Introduction}\seclab{introduction}
In the streaming model of computation, one wants to compute or approximate a predetermined function on a dataset. The dataset is implicitly defined through a sequence of updates, and the goal is to use total space that is sublinear in the size of the dataset. 
The streaming model captures key resource requirements of algorithms for many database and network tasks where the size of the data is significantly larger than the available storage, such as logs for network traffic, IoT sensors, financial markets, commercial transactions, and scientific data, e.g., astronomy or bioinformatics. 

In the classical \emph{oblivious} streaming model, there exists a stream $S$ of updates $u_1,\ldots,u_m$ that defines an underlying dataset, such as a frequency vector, a graph, or a set of points in Euclidean space. 
The sequence of updates may be worst-case, but the dataset is fixed in advance and is oblivious to any algorithmic design choices. 
Although there are examples of fundamental streaming algorithms that are deterministic, many streaming algorithms crucially utilize randomness to achieve meaningful guarantees in sublinear space. 
For example, the famous AMS sketch \cite{AlonMS99} for $F_2$ estimation initializes a random sign vector $Z$, maintains $\langle Z, f \rangle$ in the stream, and outputs $\langle Z,f\rangle^2$ which is an unbiased estimator to $\|f\|_2^2$, where $f$ is the underlying frequency vector defined by the stream. 
However, the analysis demands that the randomness used to generate the sign vector $Z$ is independent of the frequency vector $f$ and in general, the analysis of many randomized algorithms assumes that the randomness of the algorithm is independent of the input. 
However, such an assumption may not be reasonable~\cite{MironovNS11,GilbertHSWW12,BogunovicMSC17,NaorY19,AvdiukhinMYZ19,CherapanamjeriN20}; even if the stream is not adversarially generated, a user may need to repeatedly query and update a database based on the responses to previous queries. 
For example in typical optimization procedures such as stochastic gradient descent, each time step can update the eventual output by an amount based on a previous query. 
In recommendation systems, a user may choose to remove some suggestions based on personal preference and then query for a new list of recommendations.

\paragraph{(Black-box) adversarial streaming model.} 
Recently, a large body of research has been devoted to studying the (black-box) adversarial streaming model as a means of modeling adversarial data streams. In the (black-box) adversarial streaming model~\cite{Ben-EliezerY20,Ben-EliezerJWY21,HassidimKMMS20,WoodruffZ21,AlonBDMNY21,KaplanMNS21,BravermanHMSSZ21,MenuhinN21,AttiasCSS21,Ben-EliezerEO22,ChakrabartiGS22}, the sequence of updates $u_1,\ldots,u_m$ is chosen adaptively rather than being fixed. 
In particular, the input is chosen by an adversary who repeatedly queries the streaming algorithm for a fixed property of the underlying dataset at each time $t\in[m]$ and determines the update $u_{t+1}$ only after seeing the output of the algorithm after time $t$. The streaming algorithm must still be correct at all times. 
In the black-box adversarial streaming model, \cite{Ben-EliezerY20,AlonBDMNY21} show that Bernoulli sampling and reservoir sampling can approximately preserve statistics such as densities of certain subsets of the universe and \cite{BravermanHMSSZ21} shows that importance sampling can use independent public randomness to approximately solve problems such as $k$-means centering, linear regression, and graph sparsification. 

However, for other important problems such as $F_p$ moment estimation, matrix rank, or estimating the number of distinct elements in the stream, \cite{Ben-EliezerJWY21,HassidimKMMS20,WoodruffZ21,AttiasCSS21,Ben-EliezerEO22,ChakrabartiGS22} crucially use the fact that the adversary who chooses the input can only see the output of the algorithm. These algorithms essentially work by arguing that it is possible to have the output of the algorithm change only a small number of times, and so only a small amount of internal randomness is revealed, which allows such algorithms to still be correct.
However, these algorithms completely fail if the internal state of the algorithm at each point in time is also revealed to an adversary.

\paragraph{White-box adversarial streaming model.} 
In this paper, we introduce the \emph{white-box} adversarial streaming model, where the sequence of updates $u_1,\ldots,u_m$ is chosen adaptively by an adversary who sees the full internal state of the algorithm at all times, including the parameters and the previous randomness used by the algorithm. 
More formally, we define the white-box adversarial model as a two-player game between $\streamalg$, the streaming algorithm, and $\adversary$. 
Prior to the beginning of the game, a query $\calQ$ is fixed, which asks for a fixed function of some underlying dataset implicitly defined by the stream. 
The game then proceeds across $m$ rounds, where in the $t$-th round:
\begin{enumerate}
\item
$\adversary$ computes an update $u_t$ for the stream, which depends on all previous stream updates, all previous internal states of $\streamalg$, and all previous randomness used by $\streamalg$ (and thus also, all previous outputs of $\streamalg$).  
\item
$\streamalg$ uses $u_t$ to update its data structures $D_t$, acquires a fresh batch $R_t$ of random bits, and outputs a response $A_t$ to the query $\calQ$.
\item
$\adversary$ observes the response $A_t$, the internal state $D_t$ of $\streamalg$, and the random bits $R_t$.
\end{enumerate}
The goal of $\adversary$ is to make $\streamalg$ output an incorrect response $A_t$ to the query $\calQ$ at some time $t\in[m]$ throughout the stream. 
By nature of the game, only a single pass over the stream is permitted. 

\paragraph{Applications of white-box adversaries.} 
The white-box adversarial model captures the ability of an adversary to adapt to internal processes of an algorithm, which the black-box adversarial model is incapable of capturing. This property allows us to model richer adversarial scenarios. 
For example in the area of dynamic algorithms, the goal is to maintain a data structure that always outputs a correct answer at all times $t\in[m]$ across updates $u_1,\ldots,u_m$ that arrive sequentially, while minimizing either the overall running time or the worst-case update time. 
In some settings, the dynamic model also places a premium on space so that algorithms must use space sublinear in the size of the input, but generally this may not be required. 
The dynamic model often considers an adaptive adversary~\cite{Chan10,Wajc20,ChanH21,RoghaniSW22} that generates the updates $u_1,\ldots,u_m$ upon seeing the entire data structure maintained by the algorithm after the previous update, i.e., a white-box adversary. 

In general, the algorithm's internal state can be used as part of a procedure that will ultimately generate future inputs. 
For example, consider a distributed streaming setting where a centralized server wants to collect statistics on a database generated by a number of remote users. 
The centralized server wants to minimize its space usage and therefore does not want to store each update by the remote users. 
Moreover, the server may want to limit communication over the network and thus it sends components of its internal state $S$ (such as initialized random variables) to the remote users in order to optimize the information sent from the remote users back to the centralized server.  
The remote users may use $S$ in some process that ultimately affects how the data downstream is generated. 
Thus in this case, the future input data depends on (components) of the internal state $S$ of the streaming algorithm of the central coordinator; this scenario is captured by the white-box adversarial model. 
Furthermore, one of the remote users could be malicious and would like to use the state $S$ to cause the central coordinator to fail. 
In this case, the data is not only dependent on $S$ but also adversarially generated; this scenario is also captured by the white-box adversarial model. 
Finally, in the case that there is no central coordinator, the entire internal state may be stored on a cloud, which would be visible to all users. 

The pan-private streaming model~\cite{DworkNPRY10} lets the internal state of the algorithm be partially or completely revealed. 
This model is often motivated by distributed users such as hospitals, government agencies, or search engine providers. 
\cite{DworkNPRY10} notes that any data curator ``can be pressured to permit data to be used for purposes other than that for which they are collected'', including uses that may ultimately affect the distribution of future input data to the curator. 
In fact, \cite{MirMNW11} specifically considers the problem of counting distinct elements and detecting heavy-hitters on a data stream when the internal state of the algorithm is revealed, giving the motivating example of an insider manipulating traffic flow while trying to find flaws in a systems administration database that tracks user visit statistics. 
It could be argued that although the goal of \cite{MirMNW11} is just to preserve the privacy of the users, they should also consider the white-box adversarial model where the future inputs depend on previous information rather than their assumption that the input is independent of the internal information.

In machine learning, robust algorithms and adversarial attacks have captured the attention of recent research. 
In 2017, Google Brain organized a competition at NeurIPS 2017 for producing effective adversarial attacks, in which many of the most successful attacks used knowledge of the internal algorithmic parameters and training weights to minimize some loss function in a small neighborhood around the original input~\cite{BiggioCMNSLGR13,SzegedyZSBEGF14,GoodfellowSS14}. 
More recently, white-box attacks have generated adversarial inputs by modifying existing data in such a minor way that is almost imperceptible to the human eye, either in images~\cite{ SzegedyZSBEGF14,HuangPGDA17} or in the physical world~\cite{KurakinGB16,SharifBBR16,AthalyeEIK18}. 
However, the modified data results in an incorrect classification by a machine learning algorithm. 
As a result, a large body of recent literature has focused on adversarial robustness of machine learning models against white-box attacks, e.g.,~\cite{IlyasEM18,MadryMSTV18,SchmidtSTTM18,TramerKPGBM18,CubukZMVL18,KurakinGB17,LiuCLS17}.

In persistent data structures, the goal is to provide version control to an evolving data structure while minimizing either the space or time to view each version, e.g.,~\cite{DriscollSST89,FiatK03,Kaplan04}.
For example, the ability to quickly access previous versions of information stored in an online repository shared across multiple collaborators is an invaluable tool that many services already provide. 
Moreover, the internal persistent data structures used to provide version control may be accessible and thus visible to all users of the repository. 
These users may then update the persistent data structure in a manner that is not independent of the previous states.

\subsection{Our Contributions}
In this paper, we study the abilities and limitations of streaming algorithms robust to white-box adversaries, which are significantly more powerful than black-box adversaries. 
An insightful example of this is the work \cite{HardtW13}, which develops a sophisticated attack for a black-box adversary to iteratively learn the matrix used for a linear sketch in the black-box adversarial model. On the other hand, the white-box adversary immediately sees the sketching matrix when the algorithm is initiated. 
More generally, techniques such as differential privacy, which are widely employed in black-box adversarial settings to hide internal information, will not work against white-box adversaries. 

The main contributions of this paper can be summarized as introducing general tools to design algorithms robust to white-box adversaries as well as presenting a framework for proving strong lower bounds in the white-box adversarial model. In more detail, our contributions are the following:
\begin{itemize}
    \item \textbf{White-box adversarial streaming model}: We introduce and formalize the white-box adversarial model for data streams as a means of modeling richer adversarial settings found in a wide variety of application areas which are not captured by the black-box adversarial model.
    
    \item \textbf{Robust algorithms and diverse applications}: We provide streaming algorithms robust against white-box adversaries for problems across many different domains, e.g., statistical problems such as heavy-hitters, graph algorithms, applications in numerical linear algebra, and string algorithms, with wide-ranging applications. For example, estimating  $F_p$  moments has  applications  in  databases,  computer  networks,  data  mining,  and  other  contexts such as in determining data skewness, which is important in parallel database applications \cite{DeWittNSS92} or determining the output of self-joins in databases \cite{Good1989C332SI}. $L_0$ estimation is used by query  optimizers to find the number  of unique values of some attribute without having to perform an expensive sort.  This statistic is further useful for selecting a minimum-cost query plan \cite{SelingerACLP79}, database design \cite{Lightstone18}, OLAP \cite{PadmanabhanBMCH03, ShuklaDNR96}, data integration \cite{BrownHMPRS05, DasuJMS02}, data warehousing \cite{AcharyaGPR99}, and packet tracing and database auditing \cite{CormodeDIM03}. 
    For more details, see \secref{sec:intro:robust}.

    \item \textbf{Use of cryptography in robust algorithms}: A key toolkit we widely employ to design our robust algorithms comes from cryptography. Leveraging computational assumptions commonly used for the design of cryptographic protocols allows us to use powerful algorithmic tools such as collision resistant hash functions and sketching matrices for which it is computationally hard to find a ``short'' vector in their kernel.
    We believe our work opens up the possibility of using cryptography much more broadly for streaming algorithms, beyond the white-box adversarial setting.

    \item \textbf{A general lower bound framework}: We give a general reduction from two-player communication problems to space lower bounds in the white-box adversarial model. Corollaries of our reduction include lower bounds for $F_p$ moment estimation in data streams. For more details, see \secref{sec:lb:fp}.
    
    \item \textbf{Lower bounds for deterministic counting}:
    Lastly, we provide a space lower bound for deterministic algorithms which count the number of ones in a binary stream in the oblivious model, even if the algorithm has access to a timer that reports how many total stream updates it has seen so far. This lower bound serves two purposes. First, it shows that our general lower bound framework of \secref{sec:lb:fp} does not extend to multiparty (greater than two) communication protocols. Second, it provides strong lower bounds for the fundamental problem of approximately counting in a stream, which has wide applications (see \secref{sec:counting_lb}).
\end{itemize}

\subsubsection{Robust Algorithms}
\seclab{sec:intro:robust}
Any deterministic algorithm is naturally robust in the white-box adversarial streaming model, but deterministic algorithms are often inefficient in a stream. 
We first show that the $\eps$-$L_1$-heavy hitters problem can be solved using space strictly less than that of any deterministic algorithm in the white-box adversarial streaming model. 
In this problem, the goal is to output all indices $i$ such that $f_i>\eps\,L_1$, where $L_1=\|f\|_1$ is the $\ell_1$ norm of the underlying freqeuncy vector defined by the stream. 

\begin{restatable}{theorem}{thmlonehh}
\thmlab{thm:lone:hh}
There exists a white-box adversarially robust algorithm that reports all $\eps$-$L_1$-heavy hitters with probability at least $3/4$ and uses space $\O{\frac{1}{\eps}\left(\log n+\log\frac{1}{\eps}\right)+\log\log m}$. 
\end{restatable}

By comparison, the well-known Misra-Gries data structure~\cite{MisraG82} is a deterministic algorithm and thus robust against white-box adversaries, but uses $\O{\frac{1}{\eps}(\log m+\log n)}$ bits of space. 
Our algorithm offers similar guarantees to Misra-Gries in the sense that it returns a list of $\O{\frac{1}{\eps}}$ items that contains all $\eps$-$L_1$-heavy hitters as well as an approximate frequency for each item in the list with additive error $\eps\,L_1$. 
Note that in the $\eps$-$L_1$-heavy hitters problem, the gap between the frequencies of items that appear in the list can be as large as $\Omega(\eps)\cdot L_1$. 
On the other hand, the $(\varphi,\eps)$-$L_1$-heavy hitter problem demands that we find all $\varphi$-$L_1$-heavy hitters, but report no items whose frequency is below $(\varphi-\eps)\,L_1$, thereby parametrizing the threshold of the ``false positives'' that are reported in the list. 
We give an algorithm with the following guarantees for the $(\varphi,\eps)$-$L_1$-heavy hitter problem against white-box adversaries with total runtime $T$, i.e., $T$-time bounded adversaries: 
\begin{restatable}{theorem}{thmphiepshh}
\thmlab{thm:phi:eps:hh}
There exists an algorithm robust against white-box $T$-time bounded adversaries that solves the $(\varphi,\eps)$-$L_1$-heavy hitter problem with probability at least $3/4$ and uses total space \[\O{\frac{1}{\eps}\left(\log\log n+\log\frac{1}{\eps}\right)+\frac{1}{\varphi}\log n+\log\log m}\]
for $T<\poly\left(\log n,\frac{1}{\eps}\right)$ and total space \[\O{\frac{1}{\eps}\min(\log n,\log T)+\frac{1}{\varphi}\log n+\log\log m}\]
for $T\ge\poly\left(\log n,\frac{1}{\eps}\right)$. 
\end{restatable}

We remark that the proof of \thmref{thm:phi:eps:hh} uses collision-resistant hash functions and hence only guarantees robustness against white-box adversaries with polynomially bounded computation time.  
\thmref{thm:phi:eps:hh} is not information-theoretically secure against white-box adversaries with unbounded computation time.

We obtain qualitatively similar results to that of \thmref{thm:lone:hh} for the \emph{Hierarchical Heavy Hitters} problem, which generalizes $L_1$-heavy-hitters. We also obtain similar results for the vertex neighborhood identification problem, where the task is for an algorithm to identify all vertices of a graph with identical neighborhoods, in the vertex arrival model, where each update of the stream is a vertex of the graph along with a list of all of its neighbors. 

\begin{restatable}{theorem}{thmvertexid}
\thmlab{thm:vertex:id}
There exists an algorithm robust against white-box polynomial-time adversaries that reports all vertices with identical neighborhoods with probability at least $3/4$, using space $\O{n\log n}$. 
\end{restatable}
We also prove a lower bound for deterministic algorithms for the vertex neighborhood identification problem, even on oblivious data streams: 
\begin{restatable}{theorem}{thmdetvidlb}
\thmlab{thm:det:vid:lb}
Any deterministic algorithm that reports all vertices with identical neighborhoods uses space $\Omega\left(\frac{n^2}{\log n}\right)$. 
\end{restatable}
Together, \thmref{thm:vertex:id} and \thmref{thm:det:vid:lb} show a strong separation between deterministic algorithms and randomized algorithms robust against white-box adversaries with polynomial runtime. 
We further utilize cryptographic tools to obtain robust algorithms for other fundamental streaming problems. 
In particular, assuming the hardness of the Short Integer Solution (SIS) problem of lattice cryptography (see \defref{def:SIS} and \thmref{thm:SIS_hardness}), we can obtain a sublinear space algorithm for the $L_0$ estimation problem on \emph{turnstile streams}, where the stream updates are allowed to be positive and negative.
\begin{restatable}{theorem}{thmlzeroub}
\thmlab{thm:lzero:ub}
Let $c\in(0,1/2)$ and assume the adversary cannot solve the SIS problem of \defref{def:SIS} with parameter $\beta_{\infty} = \poly(n)$ in \thmref{thm:SIS_hardness} for sufficiently large $n$. Then \algref{alg:l0:ub} returns a $n^{\eps}$ multiplicative approximation to the $L_0$ estimation problem on turnstile streams.  
Furthermore, the algorithm uses space $\tilde{\mathcal{O}}(n^{1-\eps+c\eps} + n^{(1+c)\eps})$. In the random oracle model, the algorithm uses space $\tilde{\mathcal{O}}(n^{1-\eps+c\eps})$.
\end{restatable}

Note that the above result also guarantees robustness against a computationally bounded adversary, formalized in \assumref{assumption:lattice}. With no such assumptions, we can provably show that approximating general $F_p$ moments up to constant factors in data streams is impossible in sublinear space  (see \thmref{thm:fp:lb} below). 
Using the same cryptographic assumption, we also obtain streaming algorithms for the rank-decision problem in linear algebra.

\begin{restatable}{theorem}{thmrankapprox}
\thmlab{thm:rank_approx}
Suppose $A$ is a $n \times n$ matrix. There exists a constant $0 < c < 1$ such that the rank decision problem for $A$ can be solved in $\tO{nk^2}$ bits of space for any $k \le n^c$ under the random oracle model assuming a computationally bounded adversary. 
\end{restatable}

Corollaries of this result include streaming algorithms for other linear algebra based applications such as computing a linearly independent basis. Lastly, we also obtain white-box robust algorithms for string pattern matching applications.


\begin{restatable}{theorem}{thmpattern}
For an input string $P$ with given period $p$, followed by a string $U$, there exists a streaming algorithm that, with probability at least $1-\frac{1}{\poly(n)}$, finds all instances of $P$ within $U$. This streaming algorithm is robust against $T$-time white-box adversaries and uses $\O{\log T}$ bits of space. 
\end{restatable}

\subsubsection{A General Reduction for Lower Bounds}
A natural question is whether there exist robust streaming algorithms against white-box adversaries for more complex problems, such as $F_p$ estimation or matrix rank. 
To that end, we give a general technique to prove lower bounds for {\it randomized} algorithms robust to white-box algorithms through two-player {\it deterministic} communication problems. 

\begin{restatable}{theorem}{thmtwopinf}
\thmlab{thm:two:p:inf}
(Informal) Suppose there exists a white-box adversarially robust streaming algorithm using $S(n,\eps)$ space that can be used to solve a one-way two-player communication game with $S(n,\eps)$ bits of communication with probability $p\in(1/2,1]$.
Then there exists a deterministic protocol for the two-player communication game using $S(n,\eps)$ bits of communication.
\end{restatable}

\begin{proof}
Over all choices of randomness, at least $p$ fraction of the possible random strings chosen by the first player is correct over all possible inputs to the second player and at least $p$ fraction of the possible random strings chosen by the second player for each input.
The first player can enumerate over all possible inputs to the second player as well as all possible random strings to select a state that uses $S(n,\eps)$ space to represent and always succeeds. 
The first player can then send the state of the algorithm to the second player, who will then update the algorithm with their input and a fixed string, thereby resulting in a deterministic protocol that solves the one-way two-player communication game with $S(n,\eps)$ bits of communication.  
\end{proof}

We remark that \thmref{thm:two:p:inf} is especially powerful because it can be used to prove lower bounds for randomized algorithms robust against white-box adversaries using reductions from deterministic communication problems, which can often have much higher communication complexity than their randomized counterparts. 
For example, the deterministic complexity of the Equality problem, in which Alice must send Bob a message to determine whether their strings $x\in\{0,1\}^n$ and $y\in\{0,1\}^n$ are equal, is $\Theta(n)$. 
However, the randomized complexity of the Equality problem is $\Theta(\log n)$ for a constant probability of success. 
Thus \thmref{thm:two:p:inf} can show significantly stronger lower bounds for randomized algorithms robust against white-box adversaries over the standard streaming model. 
In particular, we first obtain the following hardness of approximation for $F_p$ moment estimation from \thmref{thm:two:p:inf}:

\begin{theorem}
\thmlab{thm:fp:lb}
For each $p\ge 0$ and $p \ne 1$, there exists a constant $C_p>1$ such that any white-box adversarially robust algorithm that reports a $C_p$-approximation to the frequency moment $F_p$ of an underlying frequency vector with probability at least $2/3$ must use space $\Omega(n)$. 
\end{theorem}


We also obtain the following hardness of approximation for matrix rank estimation from \thmref{thm:two:p:inf}:
\begin{theorem}
\thmlab{thm:rank:lb}
There exists a constant $C>1$ such that any white-box adversarially robust algorithm that reports a $C$-approximation to the rank of an underlying matrix with probability at least $2/3$ must use space $\Omega(n)$. 
\end{theorem}

\subsubsection{Lower Bound for Deterministic Counting}
The ``basic'' problem of counting the number of ones in a binary stream, or equivalently the number of stream updates, is arguably the most fundamental streaming problem with both practical and theoretical applications. While there exists a straightforward $\O{\log n}$ space algorithm that explicitly maintains the actual count over the stream, there exist randomized algorithms, notably the Morris counters, which achieve $\O{\log \log n}$ bits. This savings is particularly meaningful when considering applications such as counting the number of visits to a popular website such as Wikipedia. In these applications, it is common to maintain many counters rather than one, and thus optimizing the space usage per counter has a measurable overall impact.

The basic counting problem is a key subroutine in many streaming problems in the oblivious model, such as $F_p$ estimation in an insertion-only stream \cite{JayaramW19}, approximate reservoir sampling \cite{GronemeierS09}, approximating the number of inversions in a permutation \cite{AjtaiJKS02}, and $L_1$-heavy hitters in insertion streams \cite{BhattacharyyaDW19}. 

We show that any deterministic problem in the oblivious model for counting must asymptotically use the same amount of space as the trivial algorithm, even when the algorithm knows the identity of the current position in the stream. That is, even if the streaming algorithm is augmented with a ``timer'' which tells it at any time how many stream updates it has seen, the algorithm still cannot approximate the number of 1's in a binary stream up to a constant factor unless it uses $\Omega(\log n)$ bits of memory. Note that having a timer is what makes our lower bound nontrivial: without a timer, with $o(\log n)$ bits of memory there are fewer than say, $n/10$ states of the algorithm, so after seeing $n/10$ 1's, the algorithm necessarily revisits a state it has seen before. Since the algorithm is deterministic, it gets stuck in a cycle and thus can produce at best a $10$-approximation. A timer will also be useful for our application, described momentarily. 


\begin{restatable}{theorem}{thmapproxcountlb}
\thmlab{thm:approx:count:lb}
Given a constant $\eps>0$, any deterministic algorithm that outputs a $(1+\eps)$-approximation to the number of ones in a length-$n$ stream of bits must use $\Omega(\log n)$ bits of space, even if the algorithm has a timer which tells it how many stream updates it has seen so far. 
\end{restatable}


Surprisingly, \thmref{thm:approx:count:lb} shows that our technique translating two-player deterministic communication lower bounds to white-box adversary lower bounds in \thmref{thm:two:p:inf} cannot extend to an arbitrary number $n$ of players. Recall that \thmref{thm:two:p:inf} implies the white-box space complexity is at least the two-player one-way deterministic communication of the underlying communication problem, which is just the maximum communication of any player (since only one player speaks). A natural question is whether the white-box space complexity is at least the maximum communication of the underlying $n$-player deterministic communication game. This is false, since in the white-box adversarial model we can count using Morris counters with $\O{\log \log n}$ bits. However, given \thmref{thm:approx:count:lb}, the maximum communication of the underlying $n$-player deterministic communication game is $\Omega(\log n)$ bits. Note that in a communication game, each player knows its identity and can behave differently than other players, and thus the assumption that the algorithm has a timer in \thmref{thm:approx:count:lb} is needed.


\subsection{Overview of our Techniques}
\paragraph{$L_1$-heavy hitters.}
To find all $\eps$-$L_1$-heavy hitters, we first recall that the well-known Misra-Gries algorithm is deterministic and 
maintains approximate frequencies to each of the possible heavy hitters using space $\O{\frac{1}{\eps}(\log m+\log n)}$, which is expensive for $m\gg 2^n$. 
Thus if we can reduce the stream length from $m$ to some $m'=\poly\left(\frac{1}{\eps},n\right)$, then we can run Misra-Gries on the smaller stream. 
If the stream length $m$ were known, we can use Bernoulli sampling on each update in the stream with probability roughly $\frac{\log n}{\eps^2m}$ to preserve the $\eps$-$L_1$-heavy hitters; this sampling probability was shown to be secure against a white-box adversary in \cite{Ben-EliezerY20}. 
Thus it remains to resolve the issue of not knowing the stream length $m$ in advance. 

A natural approach is to make a number of exponentially increasing guesses for the length of the stream $m$. 
However, not only does this approach induce a multiplicative overhead of $\O{\log n}$ in the number of simultaneous instances, corresponding to each guess for the length, but also even tracking the length $m$ of the stream exactly requires $\O{\log m}$ bits, which we would like to avoid. 
We instead observe that Morris counters are white-box adversarially robust and use them to estimate the length of the stream at all times within a constant factor with space $\O{\log\log m}$. 
Moreover, we simultaneously only maintain two guesses for the length of the stream, corresponding to increasing powers of $\left(\frac{16}{\eps}\right)$, since when an instance of the algorithm is initiated with a guess for $m$, at most $\frac{\eps}{16}\,m$ updates have been missed by the algorithm, so any items $i$ that are $\eps$-$L_1$-heavy hitters of the stream will still be $\O{\eps}$-$L_1$-heavy hitters of the stream seen by the algorithm. Similar techniques also work for the Hierarchical Heavy Hitters problem.


\paragraph{Computationally-bounded white-box adversaries.}
We can further improve our guarantees if the white-box adversary has bounded computational time through the use of collision-resistant hash functions. 
Namely, we can solve the $(\varphi,\eps)$-$L_1$-heavy hitters problem by using a collision-resistant hash function to hash the sampled items into a universe of size $\poly\left(\log n,\frac{1}{\eps},T\right)$ for white-box adversaries with computation time at most $T$.  
Similarly, we can achieve a streaming algorithm for the vertex neighborhood identification problem that is robust against polynomial-time white-box adversaries by hashing a Boolean vector representing the neighborhood of each vector into a universe of size $\poly(n,T)$. 
Thus it suffices to maintain $n$ hashes corresponding to the $n$ neighborhoods, using total space $\O{n\log nT}$. 
By comparison, we can use the OR Equality problem to show any deterministic algorithm that solves the vertex neighborhood identification problem requires $\Omega\left(\frac{n^2}{\log n}\right)$ space. 
We similarly use collision-resistant hash functions to obtain robust algorithms for problems in linear algebra and strings. 

\paragraph{Two-player communication lower bounds.}
To acquire a tool for proving lower bounds for white-box adversarially robust algorithms, we first note that a standard technique for proving lower bounds in the oblivious streaming model is to consider the randomized communication complexity of certain problems, such as Equality. 
However, the deterministic communication complexity of many of these problems can be a lot higher. 
Surprisingly, we show that randomized algorithms that are robust against white-box adversaries can be used for deterministic protocols, thereby proving significantly stronger lower bounds for white-box adversarially robust algorithms than their counterparts in the oblivious streaming model. 
In particular, if there exists a randomized algorithm that is robust against white-box adversaries, then it can be used in a two-player communication game as follows. 
The first player reads their input and creates the stream as usual. 
Rather than generating internal randomness for the algorithm, the first player notes that there exists a choice of the internal randomness such that for any possible input to the second player, the algorithm succeeds on $9/10$ of the possible strings used for randomness by the second player. 
The first player can then enumerate over all possible inputs and all possible strings used for the randomness to the second player and choose the first internal randomness for the first player such that the guarantee holds. 
The first player can then run the algorithm on the created stream with this {\it deterministic} choice of randomness and pass the state to the second player, who can now run their input over all choices of randomness for the second part of the stream, and take a majority vote. Note that this results in a deterministic protocol and thus must respect any deterministic lower bound for a communication problem. 
In particular, we can choose the communication problem to be Gap Equality (which is just the Equality problem with the promise that when the two input strings are not equal, they differ in a constant fraction of positions) problem to show both \thmref{thm:fp:lb} and \thmref{thm:rank:lb}. 
Unfortunately, we prove that this technique cannot be generalized to show lower bounds for white-box adversarially robust algorithms through multiplayer communication. 

\paragraph{Lower bounds for deterministic counting with a timer.}
A streaming algorithm is just a read-once branching program. \thmref{thm:approx:count:lb} is proven by bounding the number of counts that a single state of the branching program can correctly represent, which translates to an upper bound on the length of an interval on a worst-case stream that each state can correctly represent. We then show that there exists some $t_0$ such that after $t_0$ updates in the stream, all read-once branching programs require at least $\poly(n)$ states to approximate the number of ones in a length $n$ stream to within a constant factor. See \secref{sec:counting_lb} for a more detailed description.


\subsection{Notation}
For an integer $n>0$, $[n]$ denotes the set of integers $\{1,\ldots,n\}$. 
We use $\poly(n)$ to denote a fixed constant degree polynomial in $n$ and $\frac{1}{\poly(n)}$ to denote an arbitrary degree polynomial in $n$ that can be determined from setting constants appropriately. 
For a vector $v\in\mathbb{R}^n$, we use $v_k$ with $k\in[n]$ to denote its $k$-th coordinate. 
Given vectors $u,v\in\mathbb{R}^n$, we write their inner product as $\langle u,v\rangle=\sum_{k=1}^n u_kv_k$. 
For a string $S$, we use $S[i:j]$ to denote the substring formed from the $i$-th character of $S$ to the $j$-th character of $S$, inclusive. 
We use $S\circ T$ to denote the concatenation of a string $S$ with a string $T$. 

\section{Upper Bounds}
\seclab{sec:ub:stat}
In this section, we present streaming algorithms relevant to statistics that are robust to white-box adversaries. 

We first describe a crucial data structure for our algorithms: the Morris counter~\cite{Morris78}. This is a data structure for the approximate counting problem, where a nonnegative integer $Z=\sum_{t=1}^m u_t$ is defined through updates $u_1,\ldots,u_m$ such that $u_t\in\{0,1\}$ for each $t\in[m]$.
Given an accuracy parameter $\eps>0$, the goal of the approximate counting problem is to estimate $Z$ to within a $(1+\eps)$-approximation. 
Our first result is that Morris counters are robust in the white-box adversarial model.

\begin{restatable}{lemma}{lemmorrisrobust}
\lemlab{lem:morris:robust}
Morris counters output a $(1+\eps)$-approximation to the frequency of an item $i$ in the white-box adversarial streaming model with probability at least $1-\delta$, using total space 
\[\O{\log\log n+\log\frac{1}{\eps}+\log\log m+\log\frac{1}{\delta}}.\]
\end{restatable}

\subsection{Heavy-Hitters}
\seclab{sec:ub:stat:hh}
In this section, we present randomized algorithms for $\eps$-$L_1$-heavy hitters that provide better guarantees than the well-known deterministic Misra-Gries data structure.
In the $F_p$ moment estimation and $L_p$ norm estimation problems, an underlying frequency vector $f\in\mathbb{R}^n$ is defined through updates $u_1,\ldots,u_m$ such that $u_t\in[n]$ for each $t\in[m]$. 
The resulting frequency vector $f$ is then defined so that 
$f_k=|\{t\,|\,u_t=k\}|$ for each $k\in[n]$. 
For $p>0$, the $F_p$ moment of $f$ is defined to be $F_p(f)=\sum_{i=1}^n(f_k)^p$. 
The $L_p$ norm of $f$ is $\|f\|_p=(F_p(f))^{1/p}$. 
We use both $F_0$ and $L_0$ to denote the number of nonzero coordinates of $f$, i.e., $F_0(f)=\|f\|_0=|\{k\,|\,f_k\neq 0\}|$. 
Given a threshold parameter $\eps>0$, the goal of the $F_p$ moment estimation problem is to provide a $(1+\eps)$-approximation to $F_p(f)$ and the goal of the $\eps$-$L_p$-heavy hitters problem is to find all coordinates $k$ such that $f_k\ge\eps\,L_p(f)$. 
In the $(\varphi,\eps)$-$L_p$-heavy hitter problem, the goal is to report all coordinates $k$ such that $f_k\ge\varphi\,L_p(f)$ but also no coordinate $j$ such that $f_j\le(\varphi-\eps)\,L_p(f)$.

\begin{theorem}
\cite{MisraG82}
\thmlab{thm:misragries}
Given a threshold parameter $\eps>0$, there exists a deterministic one-pass streaming algorithm $\misragries$ that uses $\O{\frac{1}{\eps}(\log m+\log n)}$ bits of space on a stream of length $m$ and outputs a list $L$ of size $\frac{1}{\eps}$ that includes all items $i$ such that $f_i>\eps m$. 
Moreover, the algorithm returns an estimate $\widehat{f_i}$ for each $i\in L$ such that $f_i-\eps m\le\widehat{f_i}\le f_i$.  
\end{theorem}

To output the identities of $\frac{1}{\eps}$ heavy-hitters, $\Omega\left(\frac{1}{\eps}\log n\right)$ space is clearly necessary. 
On the other hand, it is not clear that the dependence on $\log m$ for Misra-Gries in \thmref{thm:misragries} is needed. 
It is known that we can essentially preserve the $\eps$-$L_1$-heavy hitters by sampling a small number of updates in a stream. 

\begin{theorem}
\cite{Ben-EliezerY20}
\thmlab{thm:bern:robust}
There exists a constant $C>0$ such that for any $\eps,\delta\in(0,1/2)$, universe size $n$, and stream length $m$, Bernoulli sampling each item of the stream with probability $p\ge\frac{C\log(n/\delta)}{\eps^2m}$ solves the heavy hitters problem with error $\eps$ in the white-box adversarial model. 
\end{theorem}

We remark that \thmref{thm:bern:robust} was proven in \cite{Ben-EliezerY20} against black-box adversaries, but their proof extends naturally to white-box adversaries because there is no additional private randomness maintained by the algorithm. 
On the other hand, \thmref{thm:bern:robust} requires that the length of the stream is known a priori, which is an assumption that we would like to remove by running multiple instances of the algorithm in parallel, with exponentially increasing guesses for the length of the stream. 
However, even to track the length of the stream requires $\O{\log m}$ bits of space. 
Instead, only an approximation to the length of the stream is required and thus we use Morris counters to remove the dependence on $\log m$. 
Finally, we observe that it suffices to maintain only two active guesses for the length of the stream at any point in time because even if we start a guess ``late'' in the stream, we will have only missed a $\poly(\eps)$-prefix length of the stream, so any $\eps$-$L_1$-heavy hitters will still be $\O{\eps}$-heavy with respect to the substream.

A useful subroutine we will need appears in \algref{alg:bern:mg}, and our full algorithm appears in \algref{alg:bernmg:full}, which gives the full guarantees of \thmref{thm:lone:hh}. 
We can also generalize this approach to hierarchical heavy hitters.

\begin{algorithm}[!htb]
\caption{$\bernmg(n,m,\eps,\delta)$}
\alglab{alg:bern:mg}
\begin{algorithmic}[1]
\Require{Universe size $n$, upper bound $m$ on the stream length, accuracy $\eps$, failure probability $\delta$, and a stream of updates $u_1,u_2,\ldots$, where each $u_i\in[n]$ represents a single update to a coordinate of the underlying vector $f$}
\Ensure{$\eps$-$L_1$-heavy hitters of the stream}
\State{Initialize an instance $\calA$ of Misra-Gries with threshold $\frac{\eps}{2}$.}
\For{each update $u_t$ with $t\in[m]$}
\State{With probability $\O{\frac{\log(n/\delta)}{\eps^2 m}}$, update $\calA$ with $u_t$}
\EndFor
\State{\Return the output of $\calA$}
\end{algorithmic}
\end{algorithm}

\begin{algorithm}[!htb]
\caption{Adversarially robust algorithm for $\eps$-$L_1$-heavy hitters}
\alglab{alg:bernmg:full}
\begin{algorithmic}[1]
\Require{Universe size $n$, accuracy $\eps$, and a stream of updates $u_1,u_2,\ldots$, where each $u_i\in[n]$ represents a single update to a coordinate of the underlying vector $f$}
\Ensure{$\eps$-$L_1$-heavy hitters of the stream}\State{Run a Morris counter that outputs a $(1+\O{\eps})$-approximation $\widehat{t}$ to the number of stream updates $t\in[m]$.}
\State{$c\gets 0$, $r\gets 2$, $\delta\gets\O{\frac{\eps}{\log m}}$}
\For{$i\in[r]$}
\State{Initialize an instance $\calA_i$ of $\bernmg(n,(16/\eps)^i,\eps/2,\delta)$}
\EndFor
\For{each update $u_t$ with $t\in[m]$}
\State{Update all instances of $\calA_i$}
\If{$\widehat{t}\ge(16/\eps)^c$}
\State{Delete $\calA_c$}
\State{$c\gets c+1$}
\State{Initialize an instance $\calA_c$ of \\ $\bernmg(n,(16/\eps)^{c+1},\eps/2,\delta)$}
\EndIf
\State{\Return the output of $\calA_c$}
\EndFor
\end{algorithmic}
\end{algorithm}

\algref{alg:bernmg:full} also solves the $(\varphi,\eps)$-$L_1$-heavy hitter problem in which all items $i$ such that $f_i\ge\varphi\|f\|_1$ are outputted and no items $j$ such that $f_j<(\varphi-\eps)\|f\|_1$ are outputted, using a total space of  $\O{\frac{1}{\eps}\left(\log n+\log\frac{1}{\eps}\right)+\log\log m}$ bits. 
\cite{BhattacharyyaDW19} showed that for oblivious streams, the $\O{\frac{1}{\eps}\log n}$ dependence is not necessary. 
Similarly, we can further improve our bounds against white-box adversaries that use $T=\poly(\kappa)$ time through the following notion of collision-resistant hash functions:
\begin{definition}[Family of Collision-Resistant Hash Functions]
A set of functions $H=\{h_i:\{0,1\}^{n_i}\to\{0,1\}^{m_i}\}_{i\in I}$ is a family of collision-resistant hash functions (CRHF) if 
\begin{itemize}
\item
(Efficient generation)
There exists a probabilistic polynomial-time algorithm $\Gen$ such that $\Gen(1^\kappa)\in I$ for all $\kappa\in\mathbb{Z}^+$. 
\item 
(Compression) $m_i<n_i$ for all $i\in I$.
\item
(Efficient evaluation)
There exists a probabilistic polynomial-time algorithm $\Eval$ such that for all $i\in I$ and $x\in\{0,1\}^n_i$, we have $\Eval(x,i)=h_i(x)$. 
\item
(Collision-resistant)
For any non-uniform probabilistic polynomial-time algorithm $\calA$, there exists a negligible function $\negl$ such that
for all security parameters $\kappa\in\mathbb{N}$,
\[\PPPr{(x_0,x_1)\gets\calA(1^\kappa,h)}{x_0\neq x_1\wedge h(x_0)=h(x_1)}\le\negl(\kappa).\]
\end{itemize}
\end{definition}
Here we use negligible function to mean a function $\negl$ such that $\negl(x)=o\left(1/x^c\right)$ for any constant $c>0$. 
There are folklore constructions of collision-resistant hash functions based on the hardness of finding the discrete logarithm of a given composite number, e.g., Theorem 7.73 in \cite{KatzL14}:
\begin{theorem}
\thmlab{thm:crhf}
Under the discrete log assumption, there exists a family of collision-resistant hash functions with $m_i=\O{\log\kappa}$ for $i=\Gen(1^\kappa)$ and uses $\O{\log\kappa}$ bits of storage.
\end{theorem}
\thmref{thm:phi:eps:hh} then follows from applying \thmref{thm:crhf} to the sampled items. 

We also remark that \algref{alg:bernmg:full} can be used to estimate the inner product of two vectors $f$ and $g$ that are implicitly defined through two streams by using the following observations:
\begin{lemma}
\lemlab{lem:ip:one}
\cite{JayaramW18}
Let $f',g'\in\mathbb{R}^n$ be unscaled uniform samples of $f$ and $g$, sampled with probability $p_f\ge\frac{s}{m_f}$ and $p_g\ge\frac{s}{m_g}$, where $s=\frac{1}{\eps^2}$. 
Then with probability at least $0.99$, we have 
\[\langle p^{-1}_f f',p^{-1}_g,g'\rangle=\langle f,g\rangle\pm\eps\|f\|_1\|g\|_1.\]
\end{lemma}

\begin{lemma}
\lemlab{lem:ip:two}
\cite{NelsonNW12}
Given $f,g\in\mathbb{R}^n$, suppose $f'$ and $g'$ are vectors that satisfy
\[\|f'-f\|_\infty\le\eps\|f\|_1,\qquad \|g'-g\|_\infty\le\eps\|g\|_1.\]
Then $\langle f',g'\rangle-\langle f,g\rangle|\le12\eps\|f\|_1\|g\|_1$.
\end{lemma}
Combining \lemref{lem:ip:one} and \lemref{lem:ip:two}, we have the following:
\begin{restatable}{corollary}{corip}
\corlab{cor:ip}
There exists a white-box adversarially robust algorithm that uses space \[\O{\frac{1}{\eps}\left(\log n+\log\frac{1}{\eps}\right)+\log\log m}\] and outputs vectors $f',g'\in\mathbb{R}^n$ such that with probability at least $3/4$,
\[|\langle f',g'\rangle-\langle f,g\rangle|\le\eps\|f\|_1\|g\|_1.\]
\end{restatable}

\subsection{Hierarchical Heavy Hitters}
We now give adversarially robust algorithms for an important generalization of the $L_1$-heavy hitters problem known as the hierarchical heavy hitters (HHH) problem. We first define the notion of a hierarchical heavy hitter.

\begin{definition}[Hierarchical Heavy Hitter \cite{cormode2003}]
Let $D$ be hierarchical domain of height $h$ over $[n]$. Let elements$(T)$ be the union of elements that are descendants of a set of prefixes $T$ of the domain hierarchy. Given a threshold $\eps$, we define the set of \textup{Hierarchical Heavy Hitters} inductively. $HHH_0$, the hierarchical heavy hitters at level zero, are simply the $\eps$-$L_1$ heavy hitters. Given a prefix $p$ at level $i$ of the hierarchy, define $F(p)$ as $\sum f(e) : e \in elements(\{p\}) \wedge e \not \in elements(\cup_{\ell= 0}^{i=1} HHH_{\ell})$. $HHH_i$ is the set of Hierarchical Heavy Hitters at level $i$, that is, the set $\{p | F(p) \ge \eps m \}$. The set of of Hierarchical Heavy Hitters, HHH, is $\cup_{i=0}^h HHH_i$.
\end{definition}

Hierarchical heavy hitters have numerous applications, ranging from real-time anomaly detection \cite{zhang2004} to DDoS detection \cite{sekar2006}. Therefore, they have been extensively studied \cite{cormode2003, estan2003, cormode2004, lin2007, cormode2008, truong2009, thaler2012, basat2017, basat2018, moraney2020}.

The problem we are interested in is to find all hierarchical heavy hitters, and their frequencies, in a data stream. However, the problem defined above cannot be be solved exactly over data streams in general. Therefore, the literature on HHHs on data streams focuses on the following approximate version of the problem.

\begin{definition}[HHH Problem]\deflab{def:HHH_Problem}
Given a data stream from a hierarchical domain $D$, a threshold $\gamma \in (0,1)$, and an error parameter $\eps \in (0, \gamma)$, the \textup{hierarchical Heavy Hitter Problem} is that of identifying prefixes $p \in D$, and estimates $f_p$ of their associated frequencies to satisfy the following conditions. 
\begin{enumerate}
    \item  \textup{accuracy}: $f_p^* - \eps m \le f_p \le f_p^*$, where $f_p^*$ is the true frequency of $p$.
    \item \textup{coverage}: All prefixes $q$ not identified as approximate HHHs have $\sum f_e^* : e \in elements(\{q\}) \wedge e \not \in elements(P) \le \gamma m$, for any supplied $\gamma \ge \eps$, where $P$ is the subset of $p$'s which are descendants of $q$.
\end{enumerate}
\end{definition}

The state of the art space bound on the HHH problem is $\O{h/\eps}$ words of space from \cite{thaler2012} where $h$ is the height of the domain $D$. Notably, their algorithm is \emph{deterministic} and hence robust against a white box adversary. In terms of bits, the total space used by the algorithm of \cite{thaler2012} is $\O{\frac{h}{\eps}(\log m + \log n)}$.

\begin{theorem}\thmlab{thm:hhh_solver}
Given threshold parameters $\eps$ and $\gamma \ge \eps$, there exists a deterministic one-pass streaming algorithm that uses $\O{\frac{h}{\eps}(\log m + \log n)}$ space on a stream of length $m$ and solves the HHH Problem according to \defref{def:HHH_Problem}.
\end{theorem}

In this section, we present randomized and adversarially robust algorithms which provide better guarantees than the deterministic algorithm of \cite{thaler2012}. We first note that the results of \cite{Ben-EliezerY20} imply the following result about solving the HHH Problem.

\begin{theorem}
\cite{Ben-EliezerY20}
\thmlab{thm:bern:robust:hhh}
There exists a constant $C>0$ such that for any $\eps,\delta\in(0,1/2)$, universe $n$, and stream length $m$, Bernoulli sampling each item of the stream with probability $p\ge\frac{C\log(n/\delta)}{\eps^2m}$ solves the HHH Problem in the white-box adversarial model. 
\end{theorem}
\begin{proof}
The proof follows from letting the set of ranges $\mathcal{R}$ be equal to the prefixes $p \in D$ in Theorem $1.2$ of \cite{Ben-EliezerY20}. The size of $\mathcal{R}$ is $\O{n}$ since $D$ is a hierarchical domain (tree) over $[n]$. 
\end{proof}

We now present the analogous versions of \algref{alg:bern:mg} and \algref{alg:bernmg:full} for the HHH problem. The qualitative difference is that we can substitute calling $\calA$ of Misra-Gries with the appropriate algorithm for HHH from \cite{thaler2012}. 

\begin{algorithm}[!htb]
\caption{$\bernhhh(n,m,\eps,\delta)$}
\alglab{alg:bern:hhh}
\begin{algorithmic}[1]
\Require{Universe size $n$, upper bound $m$ on the stream length, parameters $\eps, \phi$, failure probability $\delta$, and a stream of updates $u_1,u_2,\ldots$, where each $u_i\in[n]$ represents a single update to a coordinate of the underlying vector $f$}
\Ensure{HHHs of the stream according to \defref{def:HHH_Problem}}
\State{Initialize an instance $\calA$ of HHH algorithm from \cite{thaler2012} with threshold $\frac{\eps}{2}$.}
\For{each update $u_t$ with $t\in[m]$}
\State{With probability $\O{\frac{\log(n/\delta)}{\eps^2 m}}$, update $\calA$ with $u_t$}
\EndFor
\State{\Return the output of $\calA$}
\end{algorithmic}
\end{algorithm}

\begin{algorithm}[!htb]
\caption{Adversarially robust algorithm for the hierarchical heavy hitters problem}
\alglab{alg:bernhhh:full}
\begin{algorithmic}[1]
\Require{Universe size $n$, accuracy $\eps$, and a stream of updates $u_1,u_2,\ldots$, where each $u_i\in[n]$ represents a single update to a coordinate of the underlying vector $f$}
\Ensure{Solution to HHH Problem according to \defref{def:HHH_Problem}}
\State{Run a Morris counter that outputs a $(1+\O{\eps})$-approximation $\widehat{t}$ to the number of stream updates $t\in[m]$.}
\State{$c\gets 0$, $r\gets 2$, $\delta\gets\O{\frac{\eps}{\log m}}$}
\For{$i\in[r]$}
\State{Initialize an instance $\calA_i$ of $\bernhhh(n,(16/\eps)^i,\eps/2,\delta)$}
\EndFor
\For{each update $u_t$ with $t\in[m]$}
\State{Update all instances of $\calA_i$}
\If{$\widehat{t}\ge(16/\eps)^c$}
\State{Delete $\calA_c$}
\State{$c\gets c+1$}
\State{Initialize an instance $\calA_c$ of $\bernhhh(n,(16/\eps)^{c+1},\eps/2,\delta)$}
\EndIf
\State{\Return the output of $\calA_c$}
\EndFor
\end{algorithmic}
\end{algorithm}

\begin{lemma}
\lemlab{lem:bern:hhh:correct}
With probability at least $1-\delta$, \algref{alg:bern:hhh} solves the HHH problem of \defref{def:HHH_Problem}.
\end{lemma}
\begin{proof}
By the robustness of Bernoulli sampling on adversarial streams, i.e., \thmref{thm:bern:robust:hhh}, we have that with probability at least $1-\delta$, simultaneously for all $p \in D$ with $f_p \ge \eps m$,
the number of instances of elements in the stream belonging to $p$ sampled by the stream as input to the instance $\calA$ of the algorithm from \cite{thaler2012} is at least $\frac{7\eps}{8}\,m$. The proof follows by the correctness of \thmref{thm:hhh_solver} with threshold $\frac{\eps}{2}$.
\end{proof}

\begin{restatable}{theorem}{thmlonehhh}
\thmlab{thm:lone:hhh}
There exists an algorithm robust against white-box adversaries that reports all with probability at least $3/4$ and uses space $\O{\frac{h}{\eps}\left(\log n+\log\frac{1}{\eps} + \log\log\log m \right)+\log\log m}$. 
\end{restatable}
\begin{proof}
The proof of correctness follows similarly to the proof of  \thmref{thm:lone:hh} and \lemref{lem:bern:hhh:correct} by noting that any $p \in D$ with $f_p \ge  \eps m$ at a time $t \in [t_i, t_{i+1}]$ will also satisfy $f_p \ge 3\eps/4$ at time $t_i$.

We now analyze the space complexity of \algref{alg:bernhhh:full}. As stated in the proof of \thmref{thm:lone:hh}, the Morris counter uses space $\O{\log\log n+\log\frac{1}{\eps}+\log\log m}$. Furthermore, each instance of $\bernhhh$ uses space 
\[\O{\frac{h}{\eps}\left(\log n+\log\frac{1}{\eps}+\log\log\log  m\right)}.\] 
Since there are at most $r=2$ such instances of $\bernhhh$, the total space is 
\[\O{\frac{h}{\eps}\left(\log n+\log\frac{1}{\eps}+\log\log\log m\right)}.\] 
Combining with the space used by Morris counters gives us the total space bound of 
\[\O{\frac{h}{\eps}\left(\log n+\log\frac{1}{\eps}+\log\log\log m\right) + \log \log m}. \qedhere\] 
\end{proof}
\subsection{$L_0$ Estimation}
\seclab{sec:ub:stat:lzero}
We now provide a streaming algorithm for the $L_0$ estimation problem, where the goal is to estimate the number of nonzero coordinates at the end of the stream. From the lower bound of \thmref{thm:fp:lb}, we cannot hope to estimate the $L_0$ norm of $f$ arbitrarily well. Surprisingly, we can attain a multiplicative approximation of $n^{\eps}$ for arbitrarily small $\eps$ even in the turnstile setting if we assume a \emph{computationally bounded} adversary, similar to assumptions made in cryptography. Our model of a computationally bounded adversary will deal with the following Short Integer Solution (SIS) problem from lattice based cryptography.

\begin{definition}[Short Integer Solution (SIS) Problem]\deflab{def:SIS}
Let $n,m,q$ be integers and let $\beta >0$. Given a uniformly  random matrix $A \in \mathbb{Z}^{n \times m}_q$, the SIS problem is to find a nonzero integer vector $z \in \mathbb{Z}^m$ such that $Az=0 \bmod q$ and $\|z\|_2 \le \beta$.
\end{definition}

Starting from Ajtai’s work \cite{ajtai96}, it is known that the SIS problem enjoys an average-case to worst-case hardness. That is, for some appropriate parameter settings, solving the SIS problem is at least as hard as approximating several fundamental lattice based cryptography problems in the worst case. 

\begin{theorem}\thmlab{thm:SIS_hardness} \cite{micciancio13}
Let $n$ and $m= \poly(n)$ be integers, let $\beta  \ge \beta_{\infty}  \ge 1$ be reals, let $Z= \{z \in \mathbb{Z}^m: \|z \|_2 \le \beta \text{ and } \|z\|_{\infty} \le \beta_{\infty}\}$, and let $q \ge \beta \cdot n^{\delta}$ for some constant $\delta > 0$.  Then solving SIS on average with non-negligible probability, and with parameters $n,m,q$ and solution set $Z\setminus \{0\}$, is  at least as hard as approximating lattice problems in the worst case on $n$-dimensional lattices to within a factor of $\gamma= \max\{1, \beta \cdot \beta_{\infty}/q\} \cdot \tO{\beta \sqrt{n}}$.
\end{theorem}

In the cryptography literature, lattice-based cryptography schemes are designed for any $\gamma$ smaller than $2^{o(n \log \log n / \log n)}$ and the best approximation currently known is for $\gamma = 2^{\O{n \log \log n / \log n}}$ via the LLL algorithm \cite{vinod15}. Improving the approximation factor to any asymptotically smaller $\gamma$ would be a major breakthrough in cryptography. Therefore, our computational assumption is the following, which implies breaking any of our algorithms would require a major cryptographic breakthrough:

\begin{assumption}\assumlab{assumption:lattice}
No polynomial-time adversary can approximate worst-case $n$-dimensional lattice problems within a $\gamma = 2^{o(n (\log \log n) / \log n)}$ factor.
\end{assumption}

In the $L_0$ streaming algorithm, we will only rely on hardness for much smaller values of $\gamma$. Our algorithm for the $L_0$ estimation problem in data streams is the following. It first considers a partition of $[n] $ into $n^{1-\eps}$ consecutive chunks each of $n^{\eps}$ coordinates. It then keeps track of $n^{1-\eps}$ vectors, one for each chunk, by multiplying the corresponding update with a sketching matrix derived from the SIS problem. Our final estimate is the number of our $n^{1-\eps}$ sketches which are nonzero when the stream ends. Note that we use the same sketching matrix $A$ on each chunk, as described below. 

\begin{algorithm}[!htb]
\caption{$\lzero(n,m,\eps)$}
\alglab{alg:l0:ub}
\begin{algorithmic}[1]
\Require{Universe size $n$, accuracy $\eps$, and a stream of updates $u_1,u_2,\ldots$, where each $u_i\in[n]$ represents a single update to a coordinate of the underlying vector $f$, and each $u_i$ is an integer}
\Ensure{$n^\eps$-multiplicative estimation of $L_0$ of $f$}
\State Consider $A \in \mathbb{Z}^{n^{c \eps} \times n^{\eps}}_q$ is a uniformly random matrix for $q = \poly(n)$ and any $1/2 > c > 0$
\State Keep track of $n^{1-\eps}$ vectors of length $n^{c\eps}$, initially all $0$ and each associated with a specific consecutive chunk of $n^{\eps}$ coordinates of $[n]$
\For{each update $u_t$ with $t\in[m]$}
\State Update the sketch vector associated with the $i$-th chunk by adding $u_t \cdot A_k$ to it, where $A_k$ is the $k$-th column of $A$, and where the stream update changes the $k$-th coordinate of the $i$-th chunk by an additive amount $u_t \in \mathbb{Z}$
\EndFor
\State{\Return the number of vectors that are nonzero}
\end{algorithmic}
\end{algorithm}

We now claim that if the final frequency vector $f$ satisfies $\|f\|_{\infty} \le \poly(n)$, then we can achieve an $n^{\eps}$ multiplicative approximation for the $L_0$ estimation problem. Furthermore, we can achieve improved sublinear space if we are working in the random oracle model of cryptography, which was introduced in the pioneering work of Bellare and Rogaway \cite{bellare93}.
In the random oracle model, we assume a publicly accessible random function which can be accessed to us \emph{and} the adversary. Each query gives a uniform random value from some output domain and repeated queries give consistent answers. The random oracle model is a well-studied model and has been used to design numerous cryptosystems \cite{bellare93, bellare96, canetti2004, koblitz2015}. In practice, one can use SHA256 as the random oracle. 

\thmlzeroub*

We remark that we only require $\gamma = \poly(n)$ for the application of \thmref{thm:SIS_hardness} to \thmref{thm:lzero:ub}. Furthermore, the algorithm also works for \textbf{turnstile} streams where the stream updates are allowed to be positive and negative. This is because we only require the final frequency vector $f$ to satisfy $\|f\|_{\infty} \le \poly(n)$ in \thmref{thm:SIS_hardness}; the signs of the entries in $f$ do not matter.

\begin{proof}
By \thmref{thm:SIS_hardness} and \assumref{assumption:lattice}, we know that the adversary cannot find any nonzero vector $x$ such that $Ax = 0$ and $\|x\|_{\infty} \le \poly(n)$. Thus if the vectors we track using $A$ equal $0$, we know that none of the coordinates in that chunk have a positive $f_i$ value at the end. Similarly if the vector is nonzero, we know that there is at least one (and at most $n^{\eps}$) coordinates associated with that chunk that have nonzero frequency value at the end. Thus on each chunk of coordinates, we make multiplicative error at most  $n^{\eps}$ and the theorem follows.

For the space bound, note that we can generate the appropriate column of $A$ on the fly via access to the random oracle (or we can store $A$ explicitly if we do not use the random oracle model). Thus the only space used is to keep track of the $n^{1-\eps}$ vectors of size $n^{c\eps}$, each associated with an $n^{\eps}$ chunk of coordinates of $[n]$. 
\end{proof}
We remark that \algref{alg:l0:ub} works for \textbf{turnstile} streams where the stream updates are allowed to be positive and negative because we only require the final frequency vector $f$ to satisfy $\|f\|_{\infty} \le \poly(n)$ in \thmref{thm:SIS_hardness} and the sign of the entries in $f$ does not matter.


\subsection{Graph Algorithms}
In this section, we consider graph algorithms in the white-box adversarial streaming model. 
We first consider the vertex neighborhood identification problem, where an underlying graph $G=(V,E)$ is defined through a sequence of updates $u_1,\ldots,u_n$. 
In the vertex arrival model, we have $|V|=n$ and each update $u_i$ corresponds to a vertex $v_i\in V$ as well as all vertices in $V$ incident to $v_i$. 
We say vertices $u,v\in V$ are incident if there exists an edge $(u,v)\in E$. 
We define the neighborhood of $v$ by the set of all vertices that are incident to $v$, i.e., $\mathcal{N}(v)=\{u\,|\,(u,v)\in E\}$. 
The goal of the vertex neighborhood identification problem is to identify vertices $u,v\in V$ such that $\mathcal{N}(u)=\mathcal{N}(v)$. 

\thmvertexid*
\begin{proof}
Since there are $n$ vertices, there can only be $n$ different neighborhoods. 
Note that each neighborhood can be represented as a binary vector of length $n$ where the $i$-th coordinate is $1$ if the $i$-th vertex is a neighbor and $0$ otherwise. 
We can thus use a collision-resistant hash function that maps each a vector into a universe of size $\poly(n)$ and with probability at least $3/4$, each distinct vertex neighborhood will be hashed to a different value. 
Hence it suffices to store the hash of the vertex neighborhood of each vertex and compare the hash values to see whether two vertices have the same neighborhood. 
Since there are $n$ vertices and each hash value uses space $\O{\log n}$, then the total space is $\O{n\log n}$ bits of space.
\end{proof}

We remark that \thmref{thm:vertex:id} is tight, given the following one-way randomized communication complexity:

\begin{theorem}[Theorem E.1 in \cite{MolinaroWY15}]
Suppose Alice is given strings $a_1,\ldots,a_n\in\{0,1\}^k$ and Bob is given strings $b_1,\ldots b_n\in\{0,1\}^k$ along with an index $i\in[n]$. 
Then Alice must send $\Omega(n\log k)$ bits of communication for Bob to determine whether $a_i=b_i$ with probability at least $3/4$. 
\end{theorem}

\begin{corollary}
Any randomized algorithm that with probability at least $3/4$, simultaneously reports all vertices with identical neighborhoods must use $\Omega(n\log n)$ bits of space.
\end{corollary}

On the other hand, if the algorithm is required to be deterministic, we can show a stronger lower bound through the following formulation of the OR Equality problem:

\begin{definition}
In the OR Equality problem, $\oreq_{n,k}$, Alice has strings $x_1,\ldots,x_k\in\{0,1\}^n$ and Bob has strings $y_1,\ldots,y_k\in\{0,1\}^n$. 
Their goal is to determine the $k$-bit string $(z_1,\ldots,z_k)$ where each $z_i=1$ if $x_i=y_i$ and $z_i=0$ otherwise, for each $i\in[k]$. 
\end{definition}

\begin{theorem}\cite{KushilevitzW09}
\thmlab{thm:oreq:cc}
For $k\le\frac{n}{\log n}$, the deterministic communication complexity of $\oreq_{n,k}$ is $\Omega(nk)$, even if $x_i=y_i$ for at most a single index $i\in[k]$. 
\end{theorem}

We now show that any deterministic algorithm that reports all vertices with identical neighborhoods uses space $\Omega\left(\frac{n^2}{\log n}\right)$. 
By comparison, \thmref{thm:vertex:id} uses space $\O{n\log n}$, thus showing a separation between deterministic algorithms and randomized algorithms against polynomially-bounded white-box adversaries.

\thmdetvidlb*
\begin{proof}
Given an instance of $\oreq_{n,k}$ with $k=\frac{n}{\log n}$, consider a graph $G$ with $3n$ vertices $u_1,\ldots,u_n,v_1,\ldots,v_n,r_1,\ldots,r_n$. 
For $i\in[k]$ and $j\in[n]$, we connect $u_i$ to $r_j$ if and only if the $j$-th coordinate of $x_i$ is $1$. 
Similarly, for $i\in[k]$ and $j\in[n]$, we connect $v_i$ to $r_j$ if and only if the $j$-th coordinate of $y_i$ is $1$. 
Thus two vertices $u_i$ and $v_i$ have the same neighborhood if and only if $x_i=y_i$. 
Hence any deterministic algorithm that reports all vertices with identical neighborhoods also solves $\oreq_{n,k}$. 
By \thmref{thm:oreq:cc} for $k=\frac{n}{\log n}$, any deterministic algorithm robust against white-box adversaries that reports all vertices with identical neighborhoods uses space $\Omega\left(\frac{n^2}{\log n}\right)$. 
\end{proof}

\subsection{Linear Algebra Algorithms}

In this section, we provide algorithms robust to white-box adversaries for problems in linear algebra. The first problem we focus on will be the Rank Decision Problem defined as follows.

\begin{problem}[Rank Decision Problem]
Given an integer $k$, and a matrix $A$, determine whether the rank of $A$ is at least $k$.
\end{problem}

We can also use the cryptographic hardness of the SIS problem to solve the rank decision problem in data streams under the random oracle model. We assume the stream is composed of integer updates to the rows of a matrix $A$ and the updates are bounded by $\poly(n)$. In the following text, we denote $\poly(n)$ to be any function bounded by $n^C$ for some constant $C$ which may change from line to line. However, this is not crucial if we assume $n$ is sufficiently large.

\thmrankapprox*

\begin{proof}
The algorithm chooses a matrix $H\in \mathbb{Z}_q^{k \times n}$ chosen from the same distribution as in the SIS problem of \defref{def:SIS}. 
We pick the $q$ parameter in \thmref{thm:SIS_hardness} to be $q \ge n^{k \log n}$. 
The algorithm always maintains $HA$, which is a $k \times n$ matrix. 
After the stream ends, it enumerates over all non-zero integer vectors $x$ with entries bounded by absolute value at most $n^k$ and checks if $HAx \equiv 0 \bmod q$. 
If such an $x$ is found, it outputs the rank is less than $k$ and otherwise says the rank is at least $k$.

We first analyze the space. Note that the entries of $H$ can be generated on the fly to compute the sketch $HA$ using the random oracle. Assuming the entries of $A$ are bounded by $\poly(n)$, the overall space required is $\tO{nk^2}$ since $q$ requires $\tO{k}$ bits to specify.

We now prove the correctness. If $A$ has rank less than $k$, then there exists a $k \times k$ submatrix of $A$ which has rank less than $k$. 
Therefore, there exists a non-zero $x$ such that $Ax = 0$ and that the entries of $x$ can be taken to be integers bounded by $\poly(n)^k$ (we can see this in a variety of ways. For example, any $k \times k $ matrix with entries bounded by $\poly(n)$ has determinant at most $\poly(n)^k$  or by considering the row echelon form of the matrix). Since $q > \poly(n)^k$, we know that $x$ is not $0 \bmod q$ and thus we will find a $x$ such that $HAx \equiv 0 \bmod q$.

Otherwise, $A$ has rank at least $k$. In this case, if we find an $x$ with $HAx \equiv 0 \bmod q$, then the adversary has found a $y = Ax$ with entries bounded by $\poly(n)^k$ such that $Hy = 0$ and $y$ is not $0 \bmod q$. This contradicts the assumption that the adversary is computationally bounded. Note that $y$ can have entries up to $\poly(n)^k$ so we can set $k$ as large as $n^c$ for a sufficiently small $c$ to satisfy our computational hardness requirements. Altogether, we output the correct answer in both cases. 
\end{proof}

\begin{remark}
Our algorithm also works in the \emph{turnstile} setting where an adversary can make positive or negative updates to the rows of $A$.
\end{remark}

\subsection{String Algorithms}
In this section, we consider string algorithms in the white-box adversarial streaming model. 
For oblivious streams, algorithms often utilize Karp-Rabin fingerprints, which is a form of polynomial identity testing that determines whether two strings $U$ and $V$ are equal. 
Given a string $U\in\{0,1\}^n$, the Karp-Rabin fingerprint of $U$ is $\sum_{i=1}^n U[i]\cdot x^i\bmod{p}$ for some sufficiently random large prime $p$ and a generator $x$. 
The Karp-Rabin fingerprint utilizes the Schwart-Zippel lemma to show that with high probability, the fingerprints of $U$ and $V$ are identical if $U=V$ and the fingerprints are different if $U\neq V$. 

Unfortunately, the Karp-Rabin fingerprint is not robust to white-box adversaries because Fermat's little theorem shows that a string with $U[i]=1$ and $U[j]=0$ for all $j\neq i$ has the same Karp-Rabin fingerprint as a string $V\in\{0,1\}^n$ with $V[i+p-1]=1$ and $V[j=0]$ for all $j\neq i+p-1$, given the same choice of $p$ and $x$. 
In other words, an adversary can use the information about the internal parameters of the Karp-Rabin fingerprint to easily generate a collision.  

For the white-box adversarial model, we can instead use collision-resistant hash functions to hash strings $U$ and $V$ and compare their hashes for equality. 
However, we require the property that the hash value of a string $U$ can be computed as the characters of $U$ arrive sequentially. 
Fortunately, the hash function in \thmref{thm:crhf} randomly selects a large prime $p$ with $\O{\log\kappa}$ bits and then randomly chooses a generator $g$ of $p$. 
Then the hash function maps $h(U)=g^U\bmod{p}$, which can be computed as characters of $U$ arrive sequentially. 
Thus, we set $h(U)$ to be the fingerprint of $U$. 

\begin{lemma}
\lemlab{lem:str:eq}
There exists a streaming algorithm that, with probability at least $1-\frac{1}{\poly(n)}$, determines whether (possibly adaptive) input strings $U$ and $V$ of length $n$ are equal and is robust against $T$-time white-box adversaries using $\O{\log\min(T,n)}$ bits of space. 
\end{lemma}

Similar to the techniques in \cite{PoratP09,CliffordFPSS16,CliffordKP19}, we can use \lemref{lem:str:eq} as a subroutine in the pattern matching problem, where the goal is to find all instances of a pattern $P$ in a text $U$. 
For simplicity, we consider the case where the period $p$ of the pattern $P$ is also given as part of the input. 
The period $p$ of a string $S$ of length $n$ is the smallest integer $\pi$ such that $S[1:n-\pi]=S[\pi+1:n]$, e.g., see~\cite{PoratP09,ErgunJS10,ErgunGSZ17,ErgunGSZ20}, as well as the more common notion that every $p$ characters of the string are the same. 
We use the following structural property about strings:

\begin{lemma}
\lemlab{lem:pattern:per}
\cite{PoratP09}
If a pattern $P$ with period $p$ matches the text $U$ at a position $i$, there cannot be a match between positions $i$ and $i+p$. 
\end{lemma}

Thus we can use the crucial fact that for our collision-resistant hash functions, the fingerprint of a string $U\circ V$ can be computed from the hashes of $U$ and $V$ and the length of $V$, to perform pattern matching in a stream. 

\begin{algorithm}[!htb]
\caption{Pattern matching in a stream}
\alglab{alg:pattern}
\begin{algorithmic}[1]
\Require{Pattern $P$ and its period $p$ of length $n$, text $U$ of length $\poly(n)$, runtime $T$ of white-box adversary}
\Ensure{All positions of $P$ in $U$}
\State{Choose a collision-resistant hash function $h$ against adversaries with runtime $\max(T,\poly(n))$.}
\State{Compute the fingerprints $\psi\gets h(P[1:p])$, $\phi\gets h(P)$.}
\State{$m\gets\emptyset$}
\For{each index $i$ such that $h(T[i+1:i+p])=\psi$}
\If{$m\not\equiv i\pmod{p}$}
\State{$m\gets i$}
\EndIf
\If{$h(T[m+1:m+n])=\phi$}
\State{Output $m$}
\State{$m\gets m+p$}
\EndIf
\EndFor
\end{algorithmic}
\end{algorithm}

\begin{lemma}
For an input string $P$ with given period $p$, followed by a string $T$, there exists a streaming algorithm that, with probability at least $1-\frac{1}{\poly(n)}$, finds all instances of $P$ within $T$ that is robust against $T$-time white-box adversaries and uses $\O{\log T}$ bits of space. 
\end{lemma}
\begin{proof}
Consider \algref{alg:pattern}. 
\lemref{lem:pattern:per} implies that if $U[i+1:i+p]$ matches $P[1:p]$ for any index $i$, then either $U[i+1:i+2p]$ matches $P[1:2p]$ or there is at most one index $j\in[i+p+1,i+2p]$ that matches $P[1:p]$. 
In other words, any sequence of matches for the first $p$ characters of $P$ must either be exactly $p$ characters apart or be more than $p$ characters apart thus eliminating the possibility of previous positions that have not already been verified as periods. 
Hence for any index $i$ of $U$ at which an instance of $P$ begins, either $i\equiv m\pmod{p}$ or $i$ will be a certificate that all positions $j=m+kp$ with integer $k\ge 0$ and $j<i$ do not match the pattern $P$. 
In the former case, $i-p$ must also be an instance of the pattern, so $m$ will be increased by $p$ to $i$ after verifying the position $i-p$. 
In the latter case, all future matches of $P[1,p]$ within the next $n$ characters of $i$ must occur at multiples of $p$ away from $i$ by \lemref{lem:pattern:per}. 
Thus in either case, \algref{alg:pattern} will report $i$. 
\end{proof}

\section{Lower Bound Techniques}
\seclab{sec:lb}
In this section, we present techniques to showing lower bounds in the white-box adversarial streaming model. 
Surprisingly, our reductions can utilize deterministic communication complexity protocols, even to show lower bounds for randomized algorithms. 
We apply our techniques to show lower bounds for $F_p$ estimation for all $p\ge 0$, including estimating the number of distinct elements for $p=0$, as well as matrix rank estimation. 

\subsection{Lower Bounds for \texorpdfstring{$F_p$}{Fp} Estimation}
\seclab{sec:lb:fp}

We illustrate \thmref{thm:two:p:inf} with a lower bound for $F_p$ moment estimation through a reduction from the following formulation of the Gap Equality problem:

\begin{definition}[Gap Equality]
In the deterministic Gap Equality problem $\detgapeq_n$, Alice receives a string $x\in\{0,1\}^n$ with $|x|=\frac{n}{2}$ and Bob receives a string $y\in\{0,1\}^n$ with $|y|=\frac{n}{2}$. 
Their goal is to use a deterministic protocol to determine whether $x=y$, given the promise that either $x=y$ or the Hamming distance $\HAM(x,y)$ satisfies $\HAM(x,y)\ge\frac{n}{10}$. 
\end{definition}
\begin{theorem}(\cite{buhrman1998quantum})
\thmlab{thm:detgapeq:cc}
The deterministic communication complexity of $\detgapeq_n$ is $\Omega(n)$. 
\end{theorem}
We use the Gap Equality problem to show lower bounds for $F_p$ moment estimation in the white-box adversarial streaming model; \thmref{thm:rank:lb} follows from a similar reduction.  
\begin{restatable}{theorem}{thmlbfp}
\thmlab{thm:lb:fp}
For any $p\ge 0$ with $p\neq 1$, there exists a constant $C_p>1$ such that any white-box adversarially robust algorithm that outputs an $C_p$-approximation to $F_p$ with probability at least $9/10$ requires $\Omega(n)$ space.
\end{restatable}
We remark that a stronger version of \thmref{thm:lb:fp} can be proved, in which the streaming algorithm can hide bits from the adversary.

\begin{proof}
First note that given vectors $u,v\in\{0,1\}^n$ with $|u|=|v|=\frac{n}{2}$ and $\HAM(u,v)\ge\frac{n}{10}$, there exists a constant $C_p$ such that
$C_p\|2u\|_p\le\|u+v\|_p$
for $p\in[0,1)$ and
$C_p\|u+v\|_p\le\|2u\|_p$
for $p>1$. 
Assume for the sake of contradiction that there exists an algorithm $\calA$ that uses $o\left(\frac{n}{2^k}\right)$ space and $k$ hidden private bits and outputs a $C_p$-approximation to $F_p$ with probability at least $9/10$ in the white-box adversarial model.
Given an instance of $\detgapeq_n$, Alice receives a string $x\in\{0,1\}^n$ with $|x|=\frac{n}{2}$ and Bob receives a string $y\in\{0,1\}^n$ with $|y|=\frac{n}{2}$. 
Alice creates a stream $S$ that induces the frequency vector $x$. 

Because Alice and Bob must solve $\detgapeq_n$ deterministically, $k$ random bits cannot be selected. 
Instead, Alice runs a separate instance of $\calA$ on $S$ for each of the $2^k$ possible realizations of the $k$ random bits. 
For the $i$-th realization of the sequence of the hidden random bits under some fixed ordering, Alice deterministically chooses a sequence $R_i$ of public random bits such that $\calA$ is correct for at least $9/10$ fraction of the possible values of $y$. 
Otherwise, if such a sequence does not exist, Alice sets the sequence $R_i$ to be the all zeros sequence. 
For each $i\in[2^k]$, Alice then runs the algorithm $\calA$ on the $i$-th realization of the sequence of the hidden random bits under some fixed ordering, the deterministic fixing of the sequence $R_i$ to use as $\calA$'s public ``random bits'', and the input $x$ to create a state $\sigma_i(x)$. 
Alice then sends the set of states $\sigma_1(x),\ldots,\sigma_{2^k}(x)$ to Bob. 

Bob creates a stream that induces the frequency vector $y$. 
For each $i\in[2^k]$, Bob takes $\sigma_i(x)$, continues running $\calA$ on the created stream, so that the underlying frequency vector is $x+y$, and queries the algorithm. 
Since a $C_p$-approximation to the norm of $x+y$ distinguishes whether $x=y$ or $\HAM(x,y)\ge\frac{n}{10}$ by the above argument, then for each $i\in[2^k]$, Bob can determine whether the $i$-th instance of $\calA$ outputs whether $x=y$ or $\HAM(x,y)\ge\frac{n}{10}$ by enumerating over all possible ``random'' strings and taking the majority output by the algorithm. 
By the correctness of $\calA$ in the white-box adversarial model, for at least $9/10$ of the possible realizations of the sequence of the hidden random bits under some fixed ordering will also be correct for \emph{all possible values} of $y$, across at least $9/10$ of the possible public random bits used by the algorithm. 
Thus at least $9/10$ fraction of the states $\sigma_i(x)$ sent by Alice, where $i\in[2^k]$, will succeed for all possible values of $y$. 
Hence, at least $9/10$ fraction of the $2^k$ outputs by Bob will be correct, allowing Bob to distinguish whether $x=y$ or $\HAM(x,y)\ge\frac{n}{10}$.

From our assumption, each instance of $\calA$ uses $o\left(\frac{n}{2^k}\right)$ space. 
Thus the states $\sigma_i(x)$ sent by Alice, where $i\in[2^k]$, use at most $o(n)$ communication, which contradicts \thmref{thm:detgapeq:cc}. 
It follows that $\calA$ uses $\Omega\left(\frac{n}{2^k}\right)$ space.
\end{proof}

\subsection{Lower Bound for Deterministic Approximate Counting}
\seclab{sec:counting_lb}
A natural question is whether the techniques of \thmref{thm:two:p:inf} extend to multiplayer communication games. 
Unfortunately, \thmref{thm:approx:count:lb} shows that the technique provably cannot generalize. 

\thmapproxcountlb*

\thmref{thm:approx:count:lb} is proven by showing that all read-once branching programs require at least $\poly(n)$ states to approximately count the number of ones in a length $n$ stream. 
Hence in a communication protocol across $n$ players where each player is given a single bit and the goal is to approximately count the number of ones held across all players, the maximum communication by a single player must be at least $\Omega(\log n)$ bits. 
This implies that our techniques which reduce hardness for white-box adversarially robust algorithms to two-player deterministic communication lower bounds \emph{cannot} be generalized to multiplayer deterministic communication lower bounds, since such a generalization would imply a space lower bound of $\Omega(\log n)$ for approximate counting, whereas Morris counters use $\O{\log\log n+\log\frac{1}{\eps}+\log\frac{1}{\delta}}$ bits of space.

We now prove \thmref{thm:approx:count:lb}. Assume that there is an algorithm which, using $s$ bits of memory, counts the number of 1's in a stream, consisting of 0's and 1's, approximately. A key result of this section is that to count up to $n$ 1's up to a factor of $1+\eps$, where $\eps > 0$ is constant, we must have $s = \Omega(\log n)$; this asymptotically matches the bound required to count it exactly. Moreover the bound that holds even if the streaming algorithm has access to a clock that keeps track of the index of the input being read (the algorithm is not charged for the memory required to store the index). Such an algorithm can be modeled as an oblivious leveled read-once branching program (also known as Ordered Binary Decision Diagram, abbreviated as OBDD) of width $2^s$ over the input alphabet $\{0,1\}$. (Without loss of generality, we let the stream be infinite and therefore the length of the OBDD is also infinite.)

More generally, we show that the lower bound applies to a larger class of counting functions called \emph{monotonic counters}.
\begin{definition}[Monotonic Counters]
Let $\Sigma$ be the input alphabet. A monotonic counter is a function $\chi : \Sigma^* \to \mathbb{N}\backslash\{0\}$ satisfying $\chi(\epsilon)=1$ for the empty string\footnote{This does not entail a loss of generality for monotonic counters that can also output 0; such exceptional sequences can be handled separately.} and $\{ \chi(\sigma a) - \chi(\sigma) : a \in \Sigma\} = \{0,1\}$, for every $\sigma \in \Sigma^*$. 
\end{definition}

In words, the counter is initialized to 1 at time 1. The counter can increase by at most 1 in each time step, and stays the same for at least one input symbol and strictly increases by 1 for at least one input symbol, thereby ensuring that the counter can assume all possible values in $\{1,2,\dots,t\}$ at time $t$. 

Fix a monotonic counter $\chi$ over an input alphabet $\Sigma$. Let $P$ be an OBDD also over $\Sigma$ and let $P(\sigma)$ denote the node reached in $P$ on input $\sigma \in \Sigma^*$. Fix any node $u$ in $P$ and let $C_u = \{\chi(\sigma) : P(\sigma)=u \text{ for some } \sigma\}$ be the nonempty set of values of the monotonic counter for input sequences that reach node $u$. 

To characterize the error in $P$'s computation, we abstractly let $\eps: \mathbb{N} \to \mathbb{R}_{\ge 0}$ be a function that represents the error in approximation. For example:
\begin{enumerate}
\item $\eps(k) = \delta k$ where $\delta>0$ is a fixed constant $\quad \Longrightarrow \quad (1+\delta)$-multiplicative approximation. 

\item $\eps(k) = (n^{\delta} - 1)k$ where $0 < \delta < 1$ is a fixed constant $\quad \Longrightarrow \quad n^{\delta}$-multiplicative approximation.

\item $\eps(k) = n^{\delta}$ where $0 \le \delta < 1$ is a fixed constant $\quad \Longrightarrow \quad n^{\delta}$-additive approximation.
\end{enumerate}
We say that a set $C$ of values is \emph{$\eps$-bound} if the deviation of its maximum value from $k$ is at most $\eps(k)$ for every $k \in C$. We say that $P$ has error $\eps$ at time $t$ if $C_u$ is $\eps$-bound for every node $u$ at time $t$.

For a node $u$, let $J_u = [\min(C_u),\max(C_u)]$ be the interval that minimally covers $C_u$, and let $I_0(t)$ be the set of intervals $J_u$ over all nodes $u$ at level $t$. Define $I(t)$ to be the set of all maximal intervals (under set inclusion) in $I_0(t)$. Note that $|I(t)|$ is a lower bound on the number of nodes in $P$ at time $t$. For this section, it suffices to consider the error in approximation induced by the intervals in $I(t)$. Namely, if $P$ has error $\eps$ at time $t$, then it implies every interval of $I(t)$ is $\eps$-bound. 

\begin{lemma} \lemlab{lem:basic-props(0)}
$I(1) = \{[1,1]\}$.
\end{lemma}
\begin{proof}
This holds because the initial value of the monotonic counter is 1.
\end{proof}

\begin{lemma} \lemlab{lem:basic-props(a)}
Let $t' \ge t \ge 1$. For every interval $[k,\ell] \in I(t)$,
there exists an interval in $I(t')$ containing $[k,\ell]$. 
\end{lemma}
\begin{proof}
The statement is trivially true for $t'=t$ and we show that it holds for $t'=t+1$; it follows by induction that also holds for all $t' \ge t$. 
Let $u$ be the node in level $t$ such that $J = J_u = [k,\ell]$ is the interval associated with node $u$. Let $\sigma$ be an input sequence such that $P(\sigma) = u$ and $\chi(\sigma) = k$. Similarly, let $\rho$ be another input sequence such that $P(\rho) = u$ and $\chi(\rho) = \ell$.

Let $a \in \Sigma$ be such that $\chi(\sigma a) = \chi(\sigma) = k$. Let $v = P(\sigma a)$ be the node reached in level $t+1$. Then $v = P(\rho a)$ as well and $\chi(\rho a) \ge \chi(\rho) = \ell$. Thus $J_v \supseteq [k,\ell]$ implying that there is an interval in $I(t+1)$ that contains $[k,\ell]$.
\end{proof}

\begin{lemma} \lemlab{lem:basic-props(b)}
Let $t \ge 1$. For every interval $[k,\ell] \in I(t)$,
there exists an interval in $I(t+1)$ containing $[k+1,\ell+1]$. 
\end{lemma}

\begin{proof}
Let $u$, $\sigma$ and $\rho$ be as in the proof of~\lemref{lem:basic-props(a)}.  
Let $a \in \Sigma$ be such that $\chi(\rho a) = \chi(\rho) + 1 = \ell+1$. Let $v = P(\rho a)$ be the node reached in level $t+1$. Then $v = P(\sigma a)$ as well and $\chi(\sigma a) \le \chi(\sigma) + 1 = k+1$. Thus $J_v \supseteq [k+1,\ell+1]$ implying that there is an interval in $I(t+1)$ that contains $[k+1,\ell+1]$.\end{proof}

Now fix any family $\{I(t)\}_{t \ge 1}$, where $I(t)$ for each ``time'' $t$ is a set of maximal intervals. The proof below depends only on the assumption the family satisfies \lemref{lem:basic-props(0)}, \lemref{lem:basic-props(a)}, and \lemref{lem:basic-props(b)} above. 
Our goal is to give a lower bound on $|I(t)|$ for some appropriate time $t$. 
We say that $k$ is \emph{present} at time $t$ if there exists an interval in $I(t)$ whose left endpoint is $k$.

\begin{lemma} \lemlab{lem:one-present}
1 is present at all times.
\end{lemma}
\begin{proof}
By~\lemref{lem:basic-props(0)}, $[1,1] \in I(1)$, and by \lemref{lem:basic-props(a)}, $I(t)$ contains an interval $[k,\ell] \supseteq [1,1]$ for each $t \ge 1$. Because $k \ge 1$ we have $k=1$. 
\end{proof}

We say that $k$ is \emph{exceptional} at time $t$ if $k$ is present at time $t$ but $k+1$ is absent at time $t+1$. We show first that if exceptional counts occur very few times then we get a lower bound $|I(t)|$ for some time $t$. Fix a time horizon $n$ and the corresponding subfamily $\{I(t)\}_{t=1}^{n+1}$. For each $h \ge 1$, let $\phi_h$ be the number of times $t \in \{1,2,\dots,n\}$ that some count $1 \le k \le h$ is exceptional at time $t$.

\begin{lemma} \lemlab{lem:excep-lb}
If $(\phi_h+1)h \le n$, then there exists $t_0 \in \{1,2,\dots,n+1\}$ such that $|I(t_0)| \ge h+1$. 
\end{lemma}

\begin{proof}
In the interval $[1,n]$, mark the times $t$ when some count $1 \le k \le h$ is exceptional at time $t$. The unmarked times can be represented as a disjoint union of $\phi_h+1$ intervals. By the pigeonhole principle, there exists one interval with size at least $h$. Let $\{t, t+1, \dots,t+h-1\}$ be such that no count $1 \le k \le h$ is exceptional in that interval. We show by induction that every count in $\{1,2,\dots,k\}$ is present at time $t+k-1$ for $1 \le k \le h+1$. That would imply for $k=h+1$ that $|I(t+h)| \ge h+1$. Since $t+h-1 \le n$, the lemma holds with $t_0 = t+h$.

The base case $k=1$ follows by \lemref{lem:one-present}. Assume the statement holds for some $1 \le k \le h$. Then, since each count in $\{1,2,\dots,k\}$ is non-exceptional at time $t+k-1$, it follows by definition that each count in $\{2,\dots, k+1\}$ is present at time $t+k$. Together with \lemref{lem:one-present}, it follows that each count in $\{1,2,\dots, k+1\}$ is present at time $t+k$ as well.
\end{proof}

We will now bound $\phi_h$ so that the above lemma can be applied for $h$ as large as possible. The idea below is to show that if a single count is exceptional too many times, then it belongs to some interval of large length, violating the approximability guarantee. 

\begin{lemma}\lemlab{lem:excep}
Suppose $k$ is exceptional at each time $t \in E$ for some nonempty set $E$. Then for all $t>\max(E)$, there exists an interval in $I(t)$ containing $k$ whose right endpoint is at least $k+|E|$.
\end{lemma}
\begin{proof}
For every $s \in E$, since $k$ is exceptional at time $s$, it is also present at time $s$ by definition. Let $t_0$ be the minimum value in $E$ and let $[k,\ell_0] \in I(t_0)$ for some $\ell_0$ certify $k$'s presence at time $t_0$. We show the following statement, denoted $P(E,t_0,\ell_0)$, by induction on $|E|$: for each $t > \max(E)$ there exists an interval in $I(t)$ containing $k$ such that its right endpoint is at least $\ell_0 + |E|$. Since $\ell_0 \ge k$, this proves the lemma.

For $|E|=1$, because $[k,\ell_0] \in I(t_0)$, by \lemref{lem:basic-props(b)}, there exists an interval $[k',\ell'] \supseteq [k+1,\ell_0+1]$ in $I(t_0+1)$. Because $k+1$ is absent at time $t_0+1$, we have $k' \le k$. By \lemref{lem:basic-props(a)}, there exists $J \in I(t)$ with $J \supseteq [k',\ell'] \supseteq [k,\ell_0+1]$. Thus, $k \in J$ and the right endpoint of $J$ is at least $\ell_0 + 1 = \ell_0  + |E|$. 

For $|E|>1$, let $t_1 \ne t_0$ denote $\max(E)$. Let $[k,\ell_1] \in I(t_1)$ for some $\ell_1$ certify $k$'s presence at time $t_1$. Let $F = E \backslash \{t_1\}$ and observe that $t_0 = \min(F)$ as well. Apply the induction hypothesis $P(F,t_0,\ell_0)$: since $t_1 > \max(F)$, there exists an interval in $I(t_1)$ containing $k$ such that its right endpoint is at least $\ell_0 + |F|$. On the other hand, $[k,\ell_1]$ is maximal in $I(t)$, therefore $\ell_1 \ge \ell_0 + |F|$.

Apply the induction hypothesis $P(\{t_1\},t_1,\ell_1)$: for every $t>t_1 = \max(E)$, there exists an interval in $I(t)$ containing $k$ such that its right endpoint is at least $\ell_1 + 1 \ge \ell_0 + |F| + 1 = \ell_0 + |E|$.
\end{proof}

Finally, we consider the situation where every element of $I(n+1)$ is $\eps$-bound. By \lemref{lem:excep}, each $k$ can be exceptional at most $\eps(k)$ times in $\{1,2,\dots,n\}$. This gives an immediate bound on $\phi_h$ namely $\phi_h \le \sum_{k=1}^h \eps(k)$. To apply \lemref{lem:excep-lb}, we find the largest $h$ such that $\left(1+\sum_{k=1}^h \eps(k)\right)h \le n$.

For example, when $\eps(k) = \delta k$, where $\delta>0$ is constant, we have $\sum_{k=1}^h \eps(k) \le \delta h(h+1)/2$. Therefore there exists a good choice of $h$ with $h = \Theta(n^{1/3})$ that yields a polynomial bound for $|I(t_0)|$. For a multiplicative approximation $n^{1-\delta}$, where $\delta>0$ is constant, we can choose $h = \Theta(n^{\delta'/3})$ for some $0 < \delta' < \delta$ and still obtain a polynomial bound for $|I(t_0)|$. Finally, for an additive approximation $n^{1-\delta}$, we can choose $h = \Theta(n^{\delta'/2})$ for some $0 < \delta' < \delta$.

The proof of \thmref{thm:approx:count:lb} follows by noting that since $|I(t_0)|=\Omega(\poly(n))$, the number of states in any deterministic approximate counting algorithm must also be $\Omega(\poly(n))$. 
Therefore, the algorithm requires at least $\Omega(\log n)$ bits of space. 

\subsection{A Communication Complexity Model for White-Box Adversaries}

We can formulate the communication game in \thmref{thm:two:p:inf} in terms of a communication matrix.
The communication matrix differs from existing two-player communication games because the protocol may not necessarily succeed against all possible inputs to the players; there can be specific inputs to the protocol that cause failure, provided that these inputs cannot be found by a $T$-time randomized algorithm.
The communication model is particularly interesting due to the equality problem and its previously discussed complexity in this model, which depends on the runtime $T$. 
On the other hand, problems like set disjointness and index do not seem to exhibit such a dependence. 
Thus we believe that this communication complexity model may be of independent interest.

Assume that there exists a streaming algorithm $\mathcal A$ robust against $T$-time white-box adversaries that can be used to compute a function $f(x,y)$ using $s$ bits of communication with probability $p$.
In the communication protocol, Alice creates a stream $S_x$ that induces the input $x$.
Alice runs $\mathcal A$ on $S_x$ and communicates the $s$-bit state of $\mathcal A$ to Bob.
Bob then continues running $\mathcal A$ on a stream $S_y$ that induces the input $y$, starting from the $s$-bit state that Bob receives from Alice.
The output $f(x,y)$ is the output of $\cA$ when run on the stream $S_x\circ S_y$, where $\circ$ denotes concatenation.

We now define a communication matrix for the randomized one-way communication protocol induced by $\mathcal A$. 
Consider a matrix $M$ whose rows are indexed by tuples $(x,r_x)$, where $x$ is Alice's input and $r_x$ is Alice's randomness, and whose columns are indexed by tuples $(y,r_y)$, where $y$ is Bob's input and $r_y$ is Bob's randomness.
The entry $M_{(x,r_x),(y,r_y)}$ of $M$ denotes the output of the two-player communication game when Alice holds $(x,r_x)$ and Bob holds $(y,r_y)$.
Since $\mathcal A$ uses $s$ space, there exists a partition of the rows of $M$ into $2^s$ parts such that if $(x,r_x)$ and $(x', r_{x'})$ are in the same part, then $M_{(x, r_x), (y, r_y)} = M_{(x', r_{x'}), (y, r_y)}$.
This corresponds to the fact that Alice sends Bob the same $s$-bit state $\te{state}(x,r_x)$ whether Alice held $(x,r_x)$ or $(x', r_{x'})$.
For each $(x, r_x)$, define
\begin{equation}\label{e1}
    p_{\te{state}(x,r_x)} = \min_y \Pr_{r_y} [M_{(x,r_x), (y,r_y)} = f(x,y)]
\end{equation}
which is the minimum probability that $\mathcal A$ outputs $f(x,y)$ over all possible inputs $y$ chosen by a white-box adversary.
Note that (\ref{e1}) is well-defined because $M_{(x, r_x), (y, r_y)} = M_{(x', r_{x'}), (y, r_y)}$ whenever $\te{state}(x,r_x) = \te{state}(x', r_{x'})$.
By the robustness of $\mathcal A$ against a white-box adversary, we have the guarantee that for all inputs $x$, $\mathbb{E}_{r_x}[p_{\te{state}(x,r_x)}] \ge p$.

We now consider the situation in which the white-box adversary $A$ is computationally bounded.
In this setting, the adversary may not be able to enumerate over all inputs $y$ and output the $y$ that minimizes $\Pr_{r_y} [M_{(x,r_x), (y,r_y)} = f(x,y)]$.
Hence a communication protocol $\mathcal A$ robust against computationally bounded white-box adversaries $A$ satisfies the weaker guarantee that
\begin{equation*}
    \mathbb{E}_{y = A(\te{state}(x,r_x))} \Pr_{r_y} [M_{(x,r_x), (y,r_y)} = f(x,y)] \ge p
\end{equation*}
for all computationally bounded adversaries $A$ and inputs $x$.

\section*{Acknowledgements}
Alec Sun was supported by the NSF Graduate Research Fellowship under Grant Nos. 1745016, and 2140739. 
David P. Woodruff and Samson Zhou were supported by a Simons Investigator Award and by the National Science Foundation under Grant No. CCF-1815840. 
This work was performed under the auspices of the U.S. Department of Energy by Lawrence Livermore National Laboratory under Contract DE-AC52-07NA27344.  LLNL-CONF-833470.
Sandeep Silwal is supported by an NSF Graduate Research Fellowship under Grant No. 1745302, NSF TRIPODS program (award DMS-2022448), and Simons Investigator Award.
\def\shortbib{0}
\bibliographystyle{alpha}
\bibliography{references}
\end{document}